\documentclass[11pt]{article}

\usepackage[a4paper,margin=1in]{geometry}
\usepackage{amsmath,amssymb,amsthm,mathtools}
\usepackage{enumitem}
\usepackage{setspace}
\usepackage{microtype}
\usepackage{booktabs,array}
\usepackage[authoryear,round]{natbib}
\usepackage[hidelinks]{hyperref}
\usepackage[nameinlink,capitalize,noabbrev]{cleveref}

\crefname{appendix}{appendix}{appendices}
\Crefname{appendix}{Appendix}{Appendices}

\usepackage{etoolbox}

\AtEndEnvironment{definition}{%
  \unskip\nobreak\hfill\(\lozenge\)%
}

\newtheorem{theorem}{Theorem}
\newtheorem{proposition}{Proposition}
\newtheorem{lemma}{Lemma}
\newtheorem{corollary}{Corollary}
\newtheorem{axiom}{Axiom}
\newtheorem*{conditionstar}{Condition A13$^*$}

\crefname{axiom}{Axiom}{Axioms}
\theoremstyle{definition}
\newtheorem{definition}{Definition}

\newtheorem{example}{Example}
\theoremstyle{remark}

\newcommand{\F}{\mathcal F}
\newcommand{\T}{\mathcal T}

\newcommand{\R}{\mathbb R}
\newcommand{\one}{\mathbf 1}

\newcommand{\E}{\mathbb E}

\title{\textbf{Implicit Smooth Ambiguity: Perception, Attitude, and Identification}}
\author{Kemal Ozbek\thanks{Department of Economics, University of Southampton, University Road, Southampton, S017 1BJ, United Kingdom. Email: mkemalozbek@gmail.com}}
\date{\today}

\begin{document}
\maketitle

\begin{abstract}
This paper develops a model of identifiable smooth ambiguity in a finite-state Anscombe--Aumann framework in which ambiguity perception and the ambiguity-attitude representation may depend on the act's endogenous certainty-equivalent utility level. The framework separates two sources of utility-level dependence that standard smooth-ambiguity models hold fixed. Neutral-transfer experiments within revealed cells recover conditional beliefs, while compensation experiments across cells recover cell weights and locally revealed ambiguity-attitude scales. These objects have distinct behavioral content. After the within-level structure is extended globally, ambiguity perception is identified at each certainty-equivalent level, whereas variation in admissible level-specific global attitude indices need not by itself reflect behaviorally identified variation in ambiguity attitude. Cross-level restrictions determine when the locally revealed attitude scales can instead be represented by one common global ambiguity-attitude index. The baseline characterization yields a doubly implicit smooth-ambiguity representation with level-specific perception and potentially level-specific global attitude indices. Requiring one common ambiguity-attitude index gives an intermediate stable-attitude model in which only perception may vary; imposing perception stability as well yields standard identifiable smooth ambiguity within the maintained full-support, neutrally separated, fixed-partition environment with at least three revealed cells. The contribution is therefore a behavioral decomposition of utility-level dependence in smooth ambiguity into separately revealed perception and attitude components.

\end{abstract}

\section{Introduction}\label{sec:intro}

A growing experimental literature suggests that ambiguity-sensitive choice is not invariant to the circumstances in which uncertainty is faced. Most directly for the present question, \citet{BaillonPlacido2019} reject constancy of ambiguity aversion as decision makers become better off overall and find substantial evidence in the direction of decreasing ambiguity aversion. Related experiments document systematic variation with monetary stakes \citep{BouchouichaEtAl2017}, across gain--loss and likelihood domains \citep{BaillonBleichrodt2015}, and in restrictions such as weak certainty independence that are retained by prominent ambiguity models \citep{TrautmannWakker2018,KonigKerstingKopsTrautmann2023}. These findings need not share a common mechanism, nor do they by themselves establish that the relevant variation is indexed by the decision maker's certainty-equivalent level. They nevertheless raise a basic identification question: when ambiguity-sensitive behavior changes as the decision maker becomes better or worse off, what exactly is changing: ambiguity perception, attitude toward ambiguity, or both?

Smooth ambiguity provides a natural setting in which to ask this question because it separates three conceptually distinct objects: a utility index for objective risk, an ambiguity-attitude index, and a belief over probabilistic models \citep{KMM2005}. In the standard formulation these objects are globally fixed. Once ambiguity-sensitive behavior is allowed to vary with utility, however, an observed reduction in ambiguity sensitivity at a higher utility level has several possible explanations. The decision maker may evaluate the same ambiguity less adversely, may attach different plausibility to the underlying probabilistic models, or may change both perception and attitude. A reduced-form model in which ambiguity evaluation varies with utility can accommodate the observation, but it cannot by itself determine which mechanism is responsible. Separating these possibilities is therefore both a modeling problem and an identification problem.

We begin from a smooth-ambiguity specification in which the way ambiguity is perceived and evaluated may vary with the decision maker's endogenous certainty-equivalent utility level. Let \(V(f)\) denote the certainty-equivalent utility of an act \(f\). At each level \(v\), let \((P_v,\mu_v)\) denote a statistically identifiable family of probabilistic models together with its second-order belief, and let \(\Phi_v\) denote the ambiguity-attitude index used at that level. The level-\(v\) threshold comparison is
\begin{equation}\label{eq:intro-di}
f\succeq\ell_v
\quad\Longleftrightarrow\quad
\sum_{p\in P_v}\mu_v(p)\,
\Phi_v\!\left(\E_p[u(f)]\right)
\geq \Phi_v(v).
\end{equation}
Both components of ambiguity evaluation may therefore vary with \(v\): the perception pair \((P_v,\mu_v)\) and the ambiguity-attitude index \(\Phi_v\). The dependence is implicit because the relevant level is not an exogenous treatment or state variable; when \(V(f)=v\), the act is evaluated using the perception and attitude associated with that same endogenous certainty-equivalent level. We refer to this as the \emph{doubly implicit} specification.

The subsequent characterizations progressively restrict these two sources of level dependence. First, they ask whether the level-specific ambiguity-attitude indices can be replaced by one common continuous strictly increasing index \(\phi\). When this is possible, ambiguity attitude is stable across certainty-equivalent levels and only ambiguity perception may remain level dependent. A final restriction asks whether perception is also invariant across levels. The analysis thus starts from a doubly implicit representation with level-specific perception and level-specific attitude indices, passes to a stable-attitude model in which only perception may vary, and finally reaches standard identifiable smooth ambiguity, in which both perception and attitude are stable, within the maintained full-support, neutrally separated fixed-partition environment with at least three revealed cells.

Finite-state identifiability makes this decomposition especially transparent. Throughout the main characterization the revealed partition $\mathcal T$ is held fixed across certainty-equivalent levels, so utility-level dependence in ambiguity perception refers to changes in the within-cell beliefs and across-cell weights conditional on a stable cell structure. Building on the finite-state implication of the identifiable-model framework of \citet{DentiPomatto2022}, a statistically identifiable perception pair $(P_v,\mu_v)$ can be represented by a partition $\mathcal T$ of the state space, a conditional belief $p_{v,A}$ within each cell $A\in\mathcal T$, and a cell-weight distribution $q_v$ with mass $q_v(A)$ assigned to cell $A$. The level-specific aggregator in \eqref{eq:intro-di} can therefore be written as
\(
\sum_{A\in\mathcal T}
q_v(A)\,
\Phi_v\!\left(
\sum_{s\in A}p_{v,A}(s)u(f(s))
\right).
\)
This representation isolates two distinct margins of changing perception: beliefs can change \emph{within} a model cell through $p_{v,A}$, while the relative plausibility assigned to different cells can change through $q_v$. The identification problem is therefore richer than asking whether a single probability distribution moves with utility.

The paper's main contribution is to recover the relevant components constructively from choice. Objective-lottery comparisons first recover the vNM utility scale and the certainty-equivalent index. At a given certainty-equivalent level, two-state neutral-transfer experiments recover within-cell probability ratios; mean-preserving transfers test whether those probabilities summarize behavior within each revealed cell; and cross-cell compensation experiments recover cell weights and the local ambiguity-attitude scale. Repeating these measurements across certainty-equivalent levels tests which recovered perception margins and locally revealed attitude classes vary with utility. The construction therefore associates distinct observable comparisons with within-cell beliefs, across-cell weights, and ambiguity attitude rather than identifying them only through the fit of an aggregate functional.

The characterization mirrors this hierarchy in behavioral terms. It first recovers the local within-level structure and then determines when that structure extends to the full act domain. The cross-level restrictions subsequently test whether the locally revealed ambiguity-attitude scales are compatible with one common global attitude index and, conditional on such stability, whether the recovered perception margins are invariant across certainty-equivalent levels. Each refinement therefore has a distinct behavioral interpretation: the first concerns within-level representation, the second attitude stability, and the third perception stability. \Cref{sec:alternative-branches} studies weaker local stability requirements and alternative orders of refinement.

The closest benchmarks clarify what this sequence adds. \citet{KMM2005} provides the benchmark decomposition between objective-risk utility, ambiguity attitude, and second-order beliefs, while \citet{DentiPomatto2022} show how identifiable subjective statistical models and the relevant perception objects can be recovered from choice. The present paper introduces an endogenous utility-level dimension into that decomposition and asks which component must vary across levels. Its contribution is therefore not simply to make the smooth-ambiguity representation more flexible, but to decompose cross-level variation behaviorally into changing perception and changing attitude. The stable-attitude model isolates the intermediate case in which perception may change while one common ambiguity-attitude index suffices, and the final invariance restriction characterizes when perception is fixed as well.

The paper thereby contributes a behavioral decomposition of utility-level-dependent ambiguity into economically interpretable components that are separately recoverable within the elicitation construction. After globalization, the perception objects are identified independently of auxiliary elicitation choices at each level, while the stable-attitude restrictions identify one common global ambiguity-attitude class. The analysis therefore shows not only when a smooth-ambiguity representation exists, but also which observable comparisons recover its partition, within-cell beliefs, across-cell weights, and ambiguity-attitude scale; which cross-level restrictions stabilize each component; and how failures of those restrictions can be interpreted component by component. Standard identifiable smooth ambiguity emerges as an empirically meaningful endpoint of these stability restrictions rather than simply as an imposed special case.

The rest of the paper is organized as follows. \Cref{sec:framework} introduces the Anscombe--Aumann environment, the doubly implicit and stable-attitude representations, and finite-state identifiability. \Cref{sec:behavioral} develops the sequential behavioral construction, the global doubly implicit baseline, the successive stability characterizations, and the identification result. \Cref{sec:extensions} studies local attitude stability, alternative orders of refinement, and extensions beyond identifiable smooth ambiguity. \Cref{sec:literature} relates the analysis to the literature, and \Cref{sec:conclusion} concludes. Proofs are collected in the appendices.

\section{Framework and target representations}\label{sec:framework}

This section introduces the common decision environment and the representation classes studied below. Objective lotteries provide a stable vNM utility scale, while ambiguity perception and  attitude may depend on the certainty-equivalent level. In the finite-state identifiable setting, perception is represented by a state partition, within-cell beliefs, and across-cell weights.

\subsection{Anscombe--Aumann acts and canonical lotteries}

Let $S=\{1,\ldots,n\}$ be a finite state space and let $Z$ be a finite prize set. For any finite set $E$, write
\(
\Delta(E):=\left\{p:E\to[0,1]:\sum_{e\in E}p(e)=1\right\},
\)
and say that $p\in\Delta(E)$ has \emph{full support} if $p(e)>0$ for every $e\in E$; equivalently, $p\in\operatorname{int}\Delta(E)$, the relative interior of the simplex. Write $L=\Delta(Z)$ for the set of objective lotteries and $\F=L^S$ for the \citet{AnscombeAumann1963} act space.

Lotteries are identified with constant acts, and mixtures of acts are defined state by state. We endow each finite-dimensional simplex with its usual Euclidean topology and $\F$ with the corresponding finite product topology; all convergence and closed-graph statements below refer to these topologies. The primitive is a preference relation $\succeq$ on $\F$; as usual, $f\sim g$ means $f\succeq g$ and $g\succeq f$, while $f\succ g$ means $f\succeq g$ and not $g\succeq f$. We also write $f\preceq g$ if $g\succeq f$ and $f\prec g$ if $g\succ f$.

We use the following notation throughout. For $z\in Z$, $\delta_z$ is the degenerate lottery on $z$. The vector $\one\in\mathbb R^S$ has every coordinate equal to one, and $e_s\in\mathbb R^S$ is the unit vector at state $s$. For $p\in\Delta(S)$ and $A\subseteq S$, put $p(A):=\sum_{s\in A}p(s)$; when $p(A)>0$, write $p(\cdot\mid A)$ for the conditional probability on $A$ and
\(
E_p[x\mid A]:=\sum_{s\in A}p(s\mid A)x_s
\; \text{and}\;
\E_p[x]:=\sum_{s\in S}p(s)x_s.
\)
Thus $\E_p[u(f)]=\sum_s p(s)u(f(s))$. 

For a set $B$ and a finite index set $I$, $B^I:=\prod_{i\in I}B$. Throughout, an expression $a\gtreqless b\Longleftrightarrow c\gtreqless d$ means that the corresponding statements hold for each of $>$, $=$, and $<$; equivalently, the same relation is used on both sides.

The behavioral conditions below imply a stable vNM utility index $u$ on $L$. Because $Z$ is finite and the preference over lotteries is nontrivial, there are a best and a worst outcome, denoted $z^+$ and $z^-$. We normalize $u(z^-)=0$ and $u(z^+)=1$. For $a\in[0,1]$, let $\ell_a=a\delta_{z^+}+(1-a)\delta_{z^-}$, so $u(\ell_a)=a$. For an act $f$, write $u(f):=(u(f(s)))_{s\in S}\in[0,1]^S$ for its statewise utility vector. For $x\in[0,1]^S$, the canonical act $f_x$ is defined by $f_x(s)=\ell_{x_s}$.
Throughout, a perturbation of a canonical utility vector is called \emph{feasible} when every resulting utility coordinate remains in $[0,1]$, so that the perturbed canonical act remains in the domain.

\subsection{Doubly implicit and stable-attitude smooth ambiguity}

We begin with a broader representation in which both ambiguity perception and ambiguity attitude may depend on the endogenous utility level; we refer to this as the doubly implicit (DI) specification. Given a normalized vNM index $u$ on $L$, a partition $\T$ of $S$, conditional beliefs $p_{v,A}\in\operatorname{int}\Delta(A)$, cell-weight distribution $q_v\in\operatorname{int}\Delta(\T)$, and continuous strictly increasing indices $\Phi_v:[0,1]\to\R$, define
\[
F_v^{\mathrm{DI}}(f)
:=
\sum_{A\in\T}q_v(A)\Phi_v\!\left(m_{v,A}(f)\right),
\qquad
m_{v,A}(f):=\sum_{s\in A}p_{v,A}(s)u(f(s)).
\]
For a utility vector $x\in[0,1]^S$, we use the same notation
\(
m_{v,A}(x):=\sum_{s\in A}p_{v,A}(s)x_s,
\)
so that the two definitions agree whenever $x_s=u(f(s))$ for every state $s$.

Throughout the identifiable representation classes below, \emph{full support} means that each conditional belief $p_{v,A}$ has full support on its own revealed cell $A$ and that $q_v$ has full support on $\T$. It does not mean that the zero extension of $p_{v,A}$ outside $A$ has full support on $S$.

\begin{definition}\label{def:diisa}
A preference admits a \emph{regular doubly implicit identifiable smooth-ambiguity representation} if there exist a normalized vNM index $u$ on $L$, a partition $\T$ of $S$, families $\{p_{v,A}\in\operatorname{int}\Delta(A)\}_{v,A}$ and $\{q_v\in\operatorname{int}\Delta(\T)\}_v$, continuous strictly increasing indices $\{\Phi_v:[0,1]\to\R\}_v$, and a continuous certainty-equivalent index $V:\F\to[0,1]$ such that, with $F_v^{\mathrm{DI}}$ defined above:
\begin{enumerate}[label=(\roman*),leftmargin=2.2em]
\item $f\succeq g$ iff $V(f)\ge V(g)$, and $V(\ell)=u(\ell)$ for every constant lottery $\ell$;
\item $V$ is strictly state-monotone in vNM utility: if $u(f(s))\ge u(g(s))$ for every $s$ and the inequality is strict in at least one state, then $V(f)>V(g)$;
\item for every $f\in\F$ and $v\in(0,1)$,
\(
V(f)\gtreqless v
\;\Longleftrightarrow\;
F_v^{\mathrm{DI}}(f)\gtreqless\Phi_v(v).
\)
\end{enumerate}
\end{definition}

The partition and vNM index are stable across levels, while the conditional beliefs \(p_{v,A}\), the cell-weight distribution \(q_v\), and the level-specific DI index \(\Phi_v\) may vary with \(v\). Two increasing attitude indices that differ by a positive affine transformation represent the same ambiguity attitude. For any continuous strictly increasing index \(\psi\), write
\( [\psi] := \{\alpha+\lambda\psi:\alpha\in\mathbb R,\ \lambda>0\} \)
for its positive-affine equivalence class. Thus \(\widetilde\psi\in[\psi]\) iff \(\widetilde\psi=\alpha+\lambda\psi\) for some \(\alpha\in\mathbb R\) and \(\lambda>0\). Throughout, statements that attitude varies or is stable refer to this equivalence class rather than to a particular normalization.

\begin{definition}\label{def:neutral-separation}
A regular doubly implicit representation in the sense of \Cref{def:diisa} is \emph{neutrally separated} if, for every $v\in(0,1)$, distinct cells $A,B$, and states $s\in A$, $t\in B$, there do not exist $r>0$ and $\eta>0$ such that, for every feasible $|\varepsilon|<\eta$,
\[q_v(A)\!\left[\Phi_v\!\left(v+p_{v,A}(s)\varepsilon\right)-\Phi_v(v)\right]+q_v(B)\!\left[\Phi_v\!\left(v-p_{v,B}(t)r\varepsilon\right)-\Phi_v(v)\right]
=0\]
\end{definition}

The restriction is substantive and is not simply a nonlinearity requirement on
\(\Phi_v\). For example, if \(\Phi_v\) is affine, then neutral separation fails:
for any distinct cells \(A,B\) and states \(s\in A\), \(t\in B\), choosing
\(
r
=
\frac{q_v(A)p_{v,A}(s)}
     {q_v(B)p_{v,B}(t)}
\)
makes the two local changes exactly offset for all sufficiently small
\(\varepsilon\) of either sign. By contrast, if
\(\Phi_v(x)=x^2\), then after any first-order cancellation the resulting
second-order term is strictly positive, so no such \(r\) can generate exact
two-sided compensation. Nonlinearity alone, however, is not sufficient:
for instance, \(\Phi_v(x)=(x-v)^3\) again permits exact proportional
compensation for a suitable \(r>0\). Thus neutral separation rules out a
specific local proportional symmetry across distinct cells, rather than
merely requiring curvature of the ambiguity-attitude index.
Thus no pair of states assigned to distinct represented cells admits the exact two-sided neutral transfer used below to reveal a cell. For the stable-attitude subclass, this definition applies with $\Phi_v=\phi$ for every $v$.

The stable-attitude specification characterized below is the subclass in which one common continuous strictly increasing ambiguity-attitude index is used at every level. Given $u$, $\T$, $\{p_{v,A}\}_{v,A}$, $\{q_v\}_v$, and such an index $\phi:[0,1]\to\R$, put
\[
F_v(f):=\sum_{A\in\T}q_v(A)\phi\!\left(m_{v,A}(f)\right).
\]

\begin{definition}\label{def:iisa}
A regular doubly implicit identifiable smooth-ambiguity representation in the sense of \Cref{def:diisa} has stable attitude if the positive-affine equivalence class of \(\Phi_v\) is independent of \(v\). Equivalently, representatives can be chosen so that there is one continuous strictly increasing \(\phi:[0,1]\to\mathbb R\) satisfying
\[ \Phi_v=\phi\;\text{for all }v\in(0,1). \]
A preference admitting such a representation is said to admit a regular implicit identifiable smooth-ambiguity representation with stable attitude.
\end{definition}

The partition $\T$, the vNM index $u$, and the ambiguity-attitude index $\phi$ are stable across $v$, whereas $p_{v,A}$ and $q_v$ may vary with $v$. Thus, in the stable-attitude subclass, utility-level-dependent perception means variation in these probability margins conditional on a fixed revealed partition. The three-way threshold conditions simultaneously determine whether an act lies above, on, or below every canonical constant. For an act with $V(f)=v\in(0,1)$, the relevant level-specific representation satisfies $F_v^{\mathrm{DI}}(f)=\Phi_v(v)$ in the doubly implicit model and $F_v(f)=\phi(v)$ in the stable-attitude model. Thus, in the implicit equation determining an act's certainty equivalent, \(V(f)\) indexes the relevant perception objects \(p_{V(f),A}\) and \(q_{V(f)}\) and, in the doubly implicit specification, the ambiguity-attitude index \(\Phi_{V(f)}\). In the stable-attitude specification, only perception remains level dependent, since the common index \(\phi\) is used at every level.
 Endpoints are ranked by $V$ itself; no $F_0$ or $F_1$ is required.

The behavioral analysis below first recovers perception and a local attitude class level by level and then imposes global behavioral calibratability to obtain the global doubly implicit baseline on the full domain. The main sequence then asks only two cross-level questions: whether one ambiguity-attitude scale governs all levels globally, and whether ambiguity perception is stable as well. The local DI model is characterized below; the distinction between local and global attitude stability, and alternative orders of refinement, are developed separately in the extensions section.

Define the \emph{predictive probability} by $\pi_v(s)=q_v(A)p_{v,A}(s)$ for $s\in A$. Then $q_v(A)=\pi_v(A)$ and $p_{v,A}=\pi_v(\cdot\mid A)$. Whenever a conditional belief $p_{v,A}\in\Delta(A)$ is treated as a model on $S$, we identify it with its zero extension outside $A$. Accordingly, define the first-order model family and second-order belief by
\(
P_v:=\{p_{v,A}:A\in\T\}
\;\text{and}\;
\mu_v(p_{v,A}):=q_v(A)=\pi_v(A).
\)
Because the cell supports are disjoint and nonempty, the models $p_{v,A}$ are distinct. Thus in either target specification the perception pair $(P_v,\mu_v)$ is equivalently represented by $(\T,\pi_v)$ in finite states. Throughout the paper, \emph{ambiguity perception} refers to this full perception object $(P_v,\mu_v)$, or equivalently $(\T,\pi_v)$ when the partition is fixed. The vector $\pi_v$ alone is called the \emph{predictive probability}, while $p_{v,A}$ and $q_v$ are the within-cell and across-cell \emph{perception margins}. This convention avoids using ``perception'' and ``predictive probability'' interchangeably.

\subsection{Finite-state identifiability}\label{sec:finite-state-identifiability}

\citet{DentiPomatto2022} study smooth ambiguity with identifiable subjective statistical models. In the finite-state setting used here, call a family $P\subseteq\Delta(S)$ \emph{identifiable} if there is a map $k:S\to P$ such that $p(\{s:k(s)=p\})=1$ for every $p\in P$. Such a map is necessarily onto: for every \(p\in P\),
\( p\bigl(\{s\in S:k(s)=p\}\bigr)=1, \)
so the set \(\{s\in S:k(s)=p\}\) is nonempty. Hence some state is mapped to every \(p\in P\). Since $S$ is finite, every identifiable $P$ is therefore finite. To avoid conflating distinct notions, throughout this paper \emph{identifiable} refers to this statistical-model property. We use \emph{recovered} for objects obtained from the primitive preference relation relative to a fixed admissible elicitation construction, and reserve \emph{behaviorally identified} for recovered objects shown to be invariant to those auxiliary elicitation choices, always up to the stated normalizations.

The following lemma records the finite-state reduction used throughout the paper. It converts an identifiable family of probabilistic models into the partition and predictive-probability objects used in the behavioral identification analysis below.

\begin{lemma}\label{lem:finite-ident}
Let $P\subseteq\Delta(S)$ and suppose there is a map $k:S\to P$ such that $p(\{s:k(s)=p\})=1$ for every $p\in P$. For $p\in P$ put $A_p=\{s:k(s)=p\}$. If $\mu\in\Delta(P)$ and $\pi=\sum_{p\in P}\mu(p)p$ has full support, then the sets $A_p$ form a partition, every $p\in P$ receives positive $\mu$-weight, and $\mu(p)=\pi(A_p)$ and $p=\pi(\cdot\mid A_p)$ for every $p\in P$. Conversely, any partition $\T$ and full-support $\pi$ generate an identifiable family by $p_A=\pi(\cdot\mid A)$ and $\mu(p_A)=\pi(A)$.
\end{lemma}

Thus, in finite states, $\{(P_v,\mu_v)\}_v$ and $\{(\T,\pi_v)\}_v$ are equivalent descriptions, up to relabelling of partition cells (equivalently, models), whenever the partition is stable.

The main characterization works throughout in the following \emph{maintained environment}: full-support identifiable representations that are neutrally separated, have a fixed revealed partition across certainty-equivalent levels, and contain at least three revealed cells. These are background restrictions used by the behavioral identification argument rather than conclusions of the characterization, and they do not define smooth ambiguity itself. With one cell, ambiguity aggregation collapses to utility-level-dependent subjective expected utility and the attitude index is behaviorally irrelevant; with affine ambiguity attitude, only the predictive probability is identified and the partition is observationally redundant. The present fixed-level compensation construction uses a third cell as a balancing coordinate; with only two cells, additional structure would be required for the same identification argument. Finally, strict state monotonicity excludes null states, whose cell membership cannot be recovered by the transfer experiments. These cases delimit where the perception--attitude decomposition is behaviorally identified.

The doubly implicit representation contains the stable-attitude model when the level-specific attitude indices are positive-affinely equivalent. If perception is stable while the $\Phi_v$'s remain level specific, it instead yields a fixed-perception doubly implicit special case. Within this environment, standard identifiable smooth ambiguity is recovered only when perception is stable and the preference admits a representation whose level-specific attitude indices belong to one common positive-affine class.

\section{Behavioral identification and representation}\label{sec:behavioral}

This section first uses axioms A1--A11 to recover the local perception and compensation objects and then adds A12, global behavioral calibratability, to obtain the global doubly implicit baseline. It subsequently imposes local cross-level attitude consistency (A13), global attitude extension and transport (A14), and perception stability (A15) before collecting the resulting identification conclusions.

\subsection{Behavioral construction and axioms}

The conditions are organized sequentially, and the construction records each object at the stage where it becomes behaviorally recoverable. Unless stated otherwise, the conditions are evaluated relative to a fixed admissible elicitation construction, with the terminology ``recovered'' and ``behaviorally identified'' used as defined in \Cref{sec:finite-state-identifiability}. Not every later condition is a primitive comparison axiom: A11(ii), A12(iii), and A14 are extendability or transport conditions imposed on cardinal objects constructed from earlier preference comparisons and are formulated so as to be invariant to the admissible numerical normalizations. The finer distinction between local and global attitude stability, together with alternative orders of refinement, is deferred to \Cref{sec:alternative-branches}.

\subsubsection{Objective risk and certainty equivalents}

For a state $s$, lottery $\ell$, and background act $f$, write $(\ell_s,f_{-s})$ for the act that yields $\ell$ at $s$ and agrees with $f$ elsewhere.

\begin{axiom}[Basic preference and objective-risk structure]\label{ax:a1}
The preference $\succeq$ is complete and transitive and is continuous in the closed-graph sense: whenever $f^n\to f$, $g^n\to g$, and $f^n\succeq g^n$ for every $n$, then $f\succeq g$. Its restriction to constant acts is nontrivial and satisfies vNM independence: for lotteries $\ell,m,n$ and $\alpha\in(0,1)$, $\ell\succeq m$ iff $\alpha\ell+(1-\alpha)n\succeq\alpha m+(1-\alpha)n$.
\end{axiom}

By the vNM expected-utility theorem (see \citet{HersteinMilnor1953}), A1 yields an affine utility representation \(u\) on objective lotteries; with the normalization in \Cref{sec:framework}, this gives the canonical lotteries \(\ell_a\)

\begin{axiom}[Statewise substitution]\label{ax:a2}
If $\ell\sim m$ as constant acts, then $(\ell_s,f_{-s})\sim(m_s,f_{-s})$ for every $s$ and every background act $f$.
\end{axiom}

Statewise substitution makes acts depend only on their statewise objective-utility consequences; the next condition imposes the strict monotonicity needed to order those consequences.

\begin{axiom}[Strict state monotonicity]\label{ax:a3}
If $f(s)\succeq g(s)$ as constant lotteries for every $s$ and the comparison is strict for at least one state, then $f\succ g$.
\end{axiom}

Under A1–A3, every act \(f\) is indifferent to the canonical act \(f_{u(f)}\) and has a unique canonical certainty equivalent \(\ell_{V(f)}\). The resulting index \(V:\mathcal F\to[0,1]\) is continuous.

\subsubsection{Revealing ambiguity perception within cells}

For $v\in(0,1)$, distinct states $s,t$, and $r>0$, define $f_{\varepsilon,r}^{v;s,t}$ by giving $\ell_{v+\varepsilon}$ in state $s$, $\ell_{v-r\varepsilon}$ in state $t$, and $\ell_v$ elsewhere. Say that $r$ is a \emph{neutral transfer rate} for $(s,t)$ at $v$ if $f_{\varepsilon,r}^{v;s,t}\sim\ell_v$ for every sufficiently small feasible positive and negative $\varepsilon$. Write $s\sim_v^0t$ when $s=t$ or such a rate exists. A3 makes the rate unique whenever $s\neq t$ and it exists; set $r_{ss}^v:=1$.

\begin{axiom}[Stable neutral-link structure and rate consistency]\label{ax:a4}
For all $v,w\in(0,1)$, $s\sim_v^0t$ iff $s\sim_w^0t$, and the resulting common relation $\sim^0$ is an equivalence relation with at least three classes. Within every class, neutral rates satisfy
\[
r_{st}^v r_{tr}^v=r_{sr}^v
\]
whenever the displayed rates are defined.
\end{axiom}

Reciprocity, $r_{st}^v r_{ts}^v=1$, need not be imposed separately. If $r_{st}^v$ is neutral for $(s,t)$, reversing the same two-sided perturbation shows that $1/r_{st}^v$ is neutral for $(t,s)$; uniqueness under A3 therefore gives $r_{ts}^v=1/r_{st}^v$.

Write $\T=S/\!\sim^0$ for the revealed partition. For each cell $A$ choose a reference state $s_A\in A$ and define
\begin{equation}\label{eq:p-from-rates}
p_{v,A}(s)=\frac{r_{s s_A}^v}{\sum_{j\in A}r_{j s_A}^v}.
\end{equation}
The multiplicative rate identities make the recovered belief independent of the chosen reference state. Indeed, if $t_A\in A$ is another reference state, then A4 gives
\[
r_{s t_A}^v=r_{s s_A}^v r_{s_A t_A}^v
\qquad\text{for every }s\in A,
\]
so replacing $s_A$ by $t_A$ multiplies every numerator and the denominator in \eqref{eq:p-from-rates} by the same positive factor $r_{s_A t_A}^v$. Moreover,
\(
r_{st}^v r_{t s_A}^v=r_{s s_A}^v,
\)
and hence $r_{st}^v=p_{v,A}(s)/p_{v,A}(t)$. Thus A4 gives a reference-state-independent conditional belief, and local neutral transfers reveal the within-cell conditional beliefs before any assumption is made about behavior away from the local neighborhood of $v$.

The three-cell restriction is used for cross-cell elicitation because one cell can then absorb the net effect of changes in two others. It is an identification requirement rather than a primitive feature of smooth ambiguity; the two-cell case is discussed in \Cref{sec:finite-state-identifiability}.

For $z=(z_A)_{A\in\T}\in[0,1]^\T$, call $z$ a \emph{block vector} and refer to \(z_A\) as the \(A\)-coordinate of \(z\). Define its associated \emph{block act} $\widehat f_z$ by $\widehat f_z(s)=\ell_{z_A}$ for $s\in A$. A block vector is \emph{$v$-indifferent} if $\widehat f_z\sim\ell_v$.

For any cell $A\in\T$ and $a,b\in[0,1]$, write $[a,b]_A$ for the tagged coordinate change $b\to a$ in cell $A$. If $I=[a,b]_A$ and $J=[c,d]_B$ are changes in distinct cells and a block vector $z$ satisfies $z_A=b$ and $z_B=c$, let $z^{I,-J}$ denote the block vector obtained by replacing $z_A$ by $a$ and $z_B$ by $d$, leaving all other coordinates unchanged. A block-coordinate change is called \emph{feasible} when every resulting block coordinate remains in $[0,1]$.

\subsubsection{Local representation and direct-measurement neighborhoods}

The local representation will be stated on an open neighborhood $J_v$ of $v$. Direct compensation experiments need only be run on a smaller neighborhood $I_v$ whose prescribed coordinates can be supported by a $v$-indifferent block act while the remaining cell means stay inside $J_v$.

\begin{lemma}\label{lem:local-solvability}
Suppose A1--A4 hold. For every $v\in(0,1)$ and every open interval $J_v\subset(0,1)$ containing $v$, there is an open interval $I_v$ with $v\in I_v$ and $\overline I_v\subset J_v$ such that, for every distinct cells $A,B$ and every $b,c\in I_v$, there is a $v$-indifferent block vector $z\in J_v^{\T}$ with $z_A=b$ and $z_B=c$.
\end{lemma}

The two neighborhoods play different roles in the elicitation argument. $J_v$ is the neighborhood on which the local smooth-ambiguity representation is required to describe choice. $I_v$ is the smaller neighborhood on which the cross-cell compensation comparisons used to measure weights and ambiguity attitude can always be implemented directly.

\begin{definition}\label{def:elicitation-construction}
An \emph{admissible elicitation construction} is a collection
\(
\bigl(\{(I_v,J_v)\}_{v\in(0,1)},\triangleleft\bigr)
\)
such that, for every $v$, $J_v\subset(0,1)$ and $I_v\subset(0,1)$ are open intervals with
\(
v\in I_v,
\;
\overline I_v\subset J_v,
\)
and, for every pair of distinct revealed cells $A,B$ and every $b,c\in I_v$, there exists a $v$-indifferent block vector $z\in J_v^{\T}$ satisfying $z_A=b$ and $z_B=c$. The relation $\triangleleft$ is a strict total order on the revealed cells used only to give a common orientation to compensation comparisons.
\end{definition}

By \Cref{lem:local-solvability}, admissible nested neighborhoods of this form exist. For the remainder of the behavioral construction fix one admissible elicitation construction. The tagged coordinate notation introduced above will be used on \(I_v\), on \(J_v\), and on the full domain \([0,1]\). All axioms from A5 onward are understood relative to this fixed choice until the characterization statements make the quantifiers explicit.

\begin{axiom}[Local within-cell mean invariance]\label{ax:a5}
Fix $v$. Suppose $f_x\sim\ell_v$, the vector of revealed cell means $m_v(x):=(m_{v,A}(x))_{A\in\T}$ belongs to $J_v^{\T}$, and $s,t$ lie in the same revealed cell. Then every feasible transfer
\(
f_{x+\varepsilon e_s-r_{st}^v\varepsilon e_t}
\)
is also indifferent to $\ell_v$.
\end{axiom}

A5 is deliberately local. It says that, within the neighborhood on which the local model is meant to apply, a transfer that leaves the revealed conditional mean of a cell unchanged does not change the comparison with $\ell_v$. It does not require the same invariance when the vector of cell means lies outside $J_v^{\T}$; A12(i) later imposes the corresponding invariance on the full domain and therefore subsumes A5.

\begin{lemma}\label{lem:within-cell-connect}
Fix a finite cell \(A\) and \(p\in\operatorname{int}\Delta(A)\). If \(x,y\in[0,1]^A\) satisfy \(p\cdot x=p\cdot y\), then there exist vectors
\(
x=x^0,x^1,\ldots,x^m=y
\)
such that, for each \(k=1,\ldots,m\), (i) \(x^k\) differs from \(x^{k-1}\) in at most two coordinates; (ii) \(p\cdot x^k=p\cdot x^{k-1}\); and (iii) for every \(s\in A\),
\(
x_s^k\in
\bigl[\min\{x_s,y_s\},\,\max\{x_s,y_s\}\bigr].
\)
Thus \(x\) can be transformed into \(y\) through finitely many two-state \(p\)-mean-preserving transfers, with each coordinate remaining between its value in \(x\) and its value in \(y\) throughout the sequence.

\end{lemma}

The connectivity result lets the local invariance condition be propagated from a single mean-preserving transfer to all acts that share the same vector of within-cell conditional means. Let
\(
m_v(f):=\bigl(m_{v,A}(f)\bigr)_{A\in\mathcal T}
\).

\begin{lemma}\label{lem:local-mean-suff}
Let \(f\) and \(g\) be two acts satisfying
\(
m_v(f)=m_v(g)\in J_v^{\mathcal T}.
\)
Under A1--A5,  for every $R\in\{\succ,\sim,\prec\}$,
\(
f\,R\,\ell_v
\;\Longleftrightarrow\;
g\,R\,\ell_v
\).
\end{lemma}

The extension from A5’s indifference invariance to preservation of the full ranking relative to \(\ell_v\) along a feasible mean-preserving transfer uses continuity: if two points on such a segment had different signs relative to \(\ell_v\), continuity would yield an indifferent intermediate act, and A5 would then make the whole segment indifferent. The appendix then applies \Cref{lem:within-cell-connect} cell by cell.

\subsubsection{Direct cross-cell compensation measurement}

A tagged change $[a,b]_A$ with $a,b\in I_v$ is called a \emph{directly measurable change} at level $v$. For two such changes $I=[a,b]_A$ and $J=[c,d]_B$ in distinct cells, \Cref{lem:local-solvability} supplies a $v$-indifferent block vector $z\in J_v^{\T}$ with any prescribed pair of starting coordinates in $I_v$. The strict total order $\triangleleft$ is used only to orient the compensation comparison below; the changes themselves are not relabelled.

\begin{axiom}[Local background invariance of direct compensation]\label{ax:a6}
For every $v$, every $A\triangleleft B$, all $a,b,c,d\in I_v$, all supporting $v$-indifferent block vectors $z,\widetilde z\in J_v^{\T}$ with $z_A=\widetilde z_A=b$ and $z_B=\widetilde z_B=c$, and every $R\in\{\succ,\sim,\prec\}$,
\(
\widehat f_{z^{I,-J}}\,R\,\ell_v
\;\Longleftrightarrow\;
\widehat f_{\widetilde z^{I,-J}}\,R\,\ell_v,
\)
where $I=[a,b]_A$ and $J=[c,d]_B$.
\end{axiom}

The direct compensation relation is defined as follows. If $I=[a,b]_A$ and $J=[c,d]_B$ with $A\triangleleft B$, set
\[
I\succeq_v^{\mathrm{dir}}J
\quad\Longleftrightarrow\quad
\widehat f_{z^{I,-J}}\succeq\ell_v
\]
for any supporting background $z$; A6 makes the definition independent of that background. If instead $B\triangleleft A$, use any supporting $v$-indifferent background $z\in J_v^{\T}$ with $z_B=d$ and $z_A=a$ and set
\[
I\succeq_v^{\mathrm{dir}}J
\quad\Longleftrightarrow\quad
\widehat f_{z^{J,-I}}\preceq\ell_v.
\]
Again A6 makes the verdict independent of the supporting background. Thus in either cell ordering $I\succeq_v^{\mathrm{dir}}J$ means that the change $I$ is at least as valuable, in compensation terms, as $J$. Recall that the compensation experiment applies \(I\) together with the reverse of \(J\), so its net represented change is the value of \(I\) minus the value of \(J\). Accordingly, exact compensation is recorded as equality of the two tagged-change values. Let $\sim_v^{\mathrm{dir}}$ and $\succ_v^{\mathrm{dir}}$ denote the symmetric and asymmetric parts. From this point through A9 the axioms concern the compensation comparisons generated from choice rather than primitive act comparisons.

A \emph{finite direct chain} from $I$ to $J$ is a finite sequence $I=K_0,\ldots,K_m=J$ of directly measurable changes, with adjacent changes in distinct cells and $K_{r-1}\succeq_v^{\mathrm{dir}}K_r$. We allow $m=0$, so every directly measurable change is chained to itself. Write $I\succeq_v^{\mathrm{ch}}J$.

For a directly measurable change \(I=[a,b]_A\), define its \emph{endpoint record} \(\chi_I\) on cell--endpoint pairs \((B,x)\in\mathcal T\times I_v\) by

\[
\chi_I(B,x)
:=
\mathbf 1\{B=A,\ x=a\}
-
\mathbf 1\{B=A,\ x=b\}.
\]

Thus \(\chi_I\) records \(+1\) at the terminal endpoint \(a\) and \(-1\) at the initial endpoint \(b\) in cell \(A\), and is zero elsewhere.
For a finite list of changes \(I_1,\ldots,I_m\), define its aggregate endpoint record by
\(
\chi_{(I_1,\ldots,I_m)}
:=
\sum_{k=1}^m \chi_{I_k},
\)
where the sum is taken pointwise on \(\mathcal T\times I_v\).

\begin{axiom}[Finite cancellation]\label{ax:a7}
For every $v$, if
\(
\chi_{(I_1,\ldots,I_m)}
=
\chi_{(J_1,\ldots,J_m)}
\) and $I_k\succeq_v^{\mathrm{dir}}J_k$ for $k=1,\ldots,m$, then \(I_k\sim_v^{\mathrm{dir}}J_k\) for every \(k=1,\ldots,m\).
\end{axiom}

A \emph{decomposition} of a directly measurable change $[a,b]_A$ is specified by a finite sequence of subdivision points
\(
b=x_0,x_1,\ldots,x_m=a,
\)
all in $I_v$, that is monotone in the direction from $b$ to $a$: weakly increasing when $a\ge b$ and weakly decreasing when $a\le b$. Its pieces are
\[
I_k=[x_k,x_{k-1}]_A,
\qquad k=1,\ldots,m.
\]
Consecutive subdivision points may coincide; when \(x_k=x_{k-1}\), the corresponding piece \([x_k,x_{k-1}]_A=[x_k,x_k]_A\) is a zero change. Two decompositions are \emph{matched} when they contain the same number of pieces.

\begin{axiom}[Finite matched decomposition]\label{ax:a8}
For every $v$ and every two directly measurable changes $I,J$, there are matched decompositions $I=(I_1,\ldots,I_m)$ and $J=(J_1,\ldots,J_m)$ such that either $I_k\succeq_v^{\mathrm{ch}}J_k$ for every $k$ or $J_k\succeq_v^{\mathrm{ch}}I_k$ for every $k$.
\end{axiom}

Define $I\succeq_v^FJ$ when some matched decompositions satisfy $I_k\succeq_v^{\mathrm{ch}}J_k$ for every $k$. Let $I\sim_v^FJ$ denote the symmetric part and $I\succ_v^FJ$ the asymmetric part. Write $- [a,b]_A=[b,a]_A$.

The following lemma records the properties of the derived compensation order that will be used in the measurement argument.

\begin{lemma}\label{lem:derived-order}
Given the fixed admissible elicitation construction, under A6--A8, $\succeq_v^F$ is a weak order on directly measurable changes. Moreover:
\begin{enumerate}[label=(\roman*),leftmargin=2.1em]
\item every direct comparison is preserved: $I\succ_v^{\mathrm{dir}}J$ implies $I\succ_v^FJ$, and $I\sim_v^{\mathrm{dir}}J$ implies $I\sim_v^FJ$;
\item all zero changes $[a,a]_A$ are mutually indifferent; denote their common $\sim_v^F$-class by $0$;
\item reversal holds: $I\succeq_v^FJ$ iff $-J\succeq_v^F-I$;
\item concatenation holds: if $[a,b]_A\succeq_v^F[a',b']_B$ and $[b,c]_A\succeq_v^F[b',c']_B$, then $[a,c]_A\succeq_v^F[a',c']_B$, with a strict conclusion if at least one premise is strict.
\end{enumerate}
\end{lemma}

For a cell $A$, a finite or infinite sequence $(x_i)_{i=1}^{N}$ in $I_v$, where $N\in\{2,3,\ldots\}\cup\{\infty\}$, is a \emph{standard sequence} if
\(
[x_{i+1},x_i]_A\sim_v^F[x_2,x_1]_A
\;\text{for every consecutive pair}
\; \text{with}\;
[x_2,x_1]_A\not\sim_v^F0.
\)
It is \emph{strictly bounded} if there are $u,w\in I_v$ such that
\(
[u,w]_A\succ_v^F[x_i,x_1]_A\succ_v^F[w,u]_A
\)
for every admissible finite index $i$ (that is, $1\le i\le N$ when $N<\infty$, and every $i\in\mathbb N$ when $N=\infty$).

\begin{axiom}[Compensation continuity, solvability, and boundedness]\label{ax:a9}
For every $v$ and cells $A,B$, the set
\(
\{(a,b,c,d)\in I_v^4:[a,b]_A\succeq_v^F[c,d]_B\}
\)
is closed relative to $I_v^4$. For every cell $A$, every $a,b\in I_v$, and every directly measurable change $I$ (in any cell), if $[a,b]_A\succeq_v^F I\succeq_v^F0$, there are $x,y\in I_v$ such that $[a,x]_A\sim_v^F I\sim_v^F[y,b]_A$. No infinite standard sequence is strictly bounded.
\end{axiom}

These order, continuity, solvability, and Archimedean properties put the derived compensation relation in the form needed for additive difference measurement.

\begin{lemma}\label{lem:additive}
Given the fixed admissible elicitation construction, under A3 and A6--A9 there are continuous strictly increasing functions $\psi_{v,A}:I_v\to\mathbb R$ such that
\begin{equation}\label{eq:local-additive}
[a,b]_A\succeq_v^F[c,d]_B
\quad\Longleftrightarrow\quad
\psi_{v,A}(a)-\psi_{v,A}(b)
\ge
\psi_{v,B}(c)-\psi_{v,B}(d).
\end{equation}
The family is unique up to cell-specific origins and one common positive unit across cells. We normalize $\psi_{v,A}(v)=0$.
\end{lemma}

The proof of \Cref{lem:additive} uses the algebraic-difference theorem of \citet{KLST1971} only through the auxiliary measurement result proved in the Appendix. Applied separately to the restriction of the derived order to each cell \(A\), that theorem yields a cell-specific function \(\varphi_A\), unique up to positive affine transformation, whose differences represent comparisons between changes within that cell:
\(
[a,b]_A\succeq_v^F[c,d]_A
\;\Longleftrightarrow\;
\varphi_A(a)-\varphi_A(b)
\ge
\varphi_A(c)-\varphi_A(d).
\)
Continuity and increasing orientation of these cell-specific scales, as well as their alignment onto the common cross-cell unit used to obtain the family \(\{\psi_{v,A}\}_A\), are established from the present assumptions in the proof of \Cref{lem:measurement-bridge}; no cross-cell cardinal comparison is imported from KLST.

\subsubsection{Local ambiguity-attitude measurement and local calibration}

For distinct cells $A,B$ and changes $I=[a,b]_A$, $J=[c,d]_B$ with endpoints in $J_v$, call the pair \emph{locally supportable} if there is a $v$-indifferent block vector $z\in J_v^{\T}$ with $z_A=b$ and $z_B=c$. Because all four endpoints lie in $J_v$, the changed vector $z^{I,-J}$ then also lies in $J_v^{\T}$. A locally supportable pair is \emph{locally exactly balanced} if $\widehat f_{z^{I,-J}}\sim\ell_v$ for some such supporting block vector. Under A11(i), the choice of supporting block vector is immaterial.

The next lemma links the measured derived order back to the primitive compensation experiment; it will be used repeatedly in the balancing arguments below.

\begin{lemma}\label{lem:derived-exact-balance}
Given the fixed admissible elicitation construction, under A6--A8, if $I$ and $J$ are directly measurable changes in distinct cells and $I\sim_v^F J$, then the pair $(I,J)$ is locally exactly balanced at level $v$.
\end{lemma}

We next require the derived difference order to be consistent across cells, so that utility differences can be measured in common units.

\begin{axiom}[Cell-unit consistency]\label{ax:a10}
For every $v$, all cells $A,B$, and all $a,b,c,d\in I_v$,
\[
[a,b]_A\succeq_v^F[c,d]_A
\quad\Longleftrightarrow\quad
[a,b]_B\succeq_v^F[c,d]_B.
\]
\end{axiom}

A10 is stated entirely in terms of the behaviorally constructed difference order. Its measurement consequence is recorded next and is then used to formulate the local calibration requirements.

\begin{lemma}\label{lem:common-scale}
Given the fixed admissible elicitation construction, under A3 and A6--A10, for every \(v\), the scales \(\{\psi_{v,A}\}_{A\in\mathcal T}\) from \(\Cref{lem:additive}\) can be normalized so that there exist a continuous strictly increasing \(\phi_v:I_v\to\mathbb R\) and numbers \(\beta_{v,A}>0\) satisfying
\(\psi_{v,A}=\beta_{v,A}\phi_v \)
for every \(A\).
\end{lemma}

By the common-unit uniqueness in \Cref{lem:additive}, the ratios $\beta_{v,A}/\beta_{v,B}$ are uniquely determined by the derived compensation order. Define the recovered normalized cell weights by
\[
q_v(A):=\frac{\beta_{v,A}}{\sum_{C\in\T}\beta_{v,C}}.
\]
Thus $q_v$ is behaviorally recovered relative to the fixed elicitation construction.

For later numerical and small-difference arguments, fix for each $v$ any nondegenerate compact interval
\[
K_v=[k_v^-,k_v^+]\subset I_v,
\qquad
v\in\operatorname{int}K_v.
\]
The core $K_v$ is an auxiliary numerical device, not part of the admissible elicitation construction, and its choice imposes no behavioral restriction. It is used only for small-difference arguments. Because $v\in\operatorname{int}K_v$ and $\phi_v$ is continuous and strictly increasing, the set of differences $\{\phi_v(a)-\phi_v(b):a,b\in K_v\}$ contains a neighborhood of zero. Thus \(K_v\) supplies the arbitrarily small difference units needed in the subsequent subdivision arguments. Extension of the level-\(v\) ambiguity-attitude scale \(\phi_v\) beyond the directly measured domain \(I_v\) is not obtained from \(K_v\); it is governed separately by the behavioral calibration requirements in A11 and A12.

Let $\mathcal E_v(\phi_v;I_v,J_v)$ be the set of continuous strictly increasing functions $\Theta:J_v\to\mathbb R$ satisfying $\Theta=\phi_v$ on $I_v$. For $\Theta\in\mathcal E_v(\phi_v;I_v,J_v)$, write $\Delta_v^\Theta(a,b)=\Theta(a)-\Theta(b)$.

\begin{axiom}[Local calibration consistency]\label{ax:a11}
For every \(v\), the following requirements hold relative to the level-$v$ ambiguity-attitude scale \(\phi_v\) recovered under A10.
\begin{enumerate}[label=(\roman*),leftmargin=2.1em]
\item \emph{Local background calibration.} For every pair of distinct cells $A,B$, every $a,b,c,d\in J_v$, all supporting $v$-indifferent block vectors $z,\widetilde z\in J_v^{\T}$ with $z_A=\widetilde z_A=b$ and $z_B=\widetilde z_B=c$, and every $R\in\{\succ,\sim,\prec\}$,
\[
\widehat f_{z^{I,-J}}\,R\,\ell_v
\quad\Longleftrightarrow\quad
\widehat f_{\widetilde z^{I,-J}}\,R\,\ell_v,
\]
where $I=[a,b]_A$ and $J=[c,d]_B$.
\item \emph{Local compatible calibration.} There exists $\Theta_v\in\mathcal E_v(\phi_v;I_v,J_v)$ such that, whenever $I=[a,b]_A$ and $J=[c,d]_B$, with $a,b,c,d\in J_v$, are locally exactly balanced and a locally supportable replacement pair $I'=[a',b']_A$ \text{and} $J'=[c',d']_B$, with $ a',b',c',d'\in J_v$, satisfies
\[
\Delta_v^{\Theta_v}(a',b')=\Delta_v^{\Theta_v}(a,b),
\qquad
\Delta_v^{\Theta_v}(c',d')=\Delta_v^{\Theta_v}(c,d),
\]
then $(I',J')$ is also locally exactly balanced.
\end{enumerate}
\end{axiom}

A11(i) deliberately subsumes A6 on directly measurable changes. A6 is retained separately because it is needed earlier to define the direct compensation relation and construct the local measurement system before A11 is imposed. A11(ii) is the substantive behavioral extendability requirement: the selected $\Theta_v$ must extend the entire directly measured scale on $I_v$ and transport local exact compensation throughout $J_v$, beyond the directly measured region.

\begin{definition}\label{def:locally-valid}
Fix $v$ and the normalization $\phi_v(v)=0$. A function \(\Theta_v:J_v\to\mathbb R\) is a \emph{locally valid ambiguity-attitude index} if it extends the recovered ambiguity-attitude scale \(\phi_v\), that is, \(\Theta_v\in\mathcal E_v(\phi_v;I_v,J_v)\), and satisfies the replacement requirement in A11(ii). Let \(\mathcal E_v^{L}\) denote the set of locally valid ambiguity-attitude indices. A11(ii) requires $\mathcal E_v^{L}\neq\varnothing$.
\end{definition}

Before the normalization at $v$ is imposed, the definition is independent of the positive-affine representative of the directly measured attitude class: if $\widetilde\phi_v=\alpha+\lambda\phi_v$ with $\lambda>0$, then $\Theta_v\in\mathcal E_v^{L}$ if and only if $\alpha+\lambda\Theta_v$ is locally valid relative to $\widetilde\phi_v$. After normalizing the representative to satisfy $\phi_v(v)=0$, the remaining freedom is multiplication by a positive constant. Hence A11(ii) restricts the positive-affine attitude class rather than an arbitrary choice of origin or units.

Thus a locally valid index is not an arbitrary continuation of the measured attitude scale: after normalization it agrees with the directly measured scale on all of $I_v$ and its differences correctly transport local exact compensation throughout the larger local representation neighborhood $J_v$. For $a,b,c,d\in I_v$, define the common same-cell difference order by
\[
[a,b]\succeq_v^\Delta[c,d]
\quad\Longleftrightarrow\quad
[a,b]_A\succeq_v^F[c,d]_A
\]
for any revealed cell $A$. By A10, the right-hand side is independent of the choice of $A$. The separation between A10 and A11 makes the logical sequence explicit: A10 recovers the common local difference scale, while A11 tests whether that recovered scale continues to organize choice on the larger local domain.

\begin{definition}\label{def:local-di}
Given the local representation neighborhoods \(\{J_v\}_{v\in(0,1)}\) from a fixed admissible elicitation construction, a preference admits a \emph{regular neutrally separated local doubly implicit representation} if there exist a normalized nonconstant vNM utility index $u$ on $L$, a continuous certainty-equivalent index $V:\F\to[0,1]$, a fixed partition $\T$ with at least three cells, full-support conditional beliefs $p_{v,A}\in\operatorname{int}\Delta(A)$, full-support cell weights $q_v\in\operatorname{int}\Delta(\T)$, and continuous strictly increasing indices $\Theta_v:J_v\to\mathbb R$ such that:
\begin{enumerate}[label=(\roman*),leftmargin=2.2em]
\item $f\succeq g$ iff $V(f)\ge V(g)$, $V(\ell)=u(\ell)$ for every constant lottery $\ell$, and acts with the same statewise $u$-values have the same $V$-value;
\item $V$ is strictly state-monotone in statewise vNM utility;
\item whenever $(m_{v,A}(f))_{A\in\T}\in J_v^{\T}$,
\begin{equation}\label{eq:local-di}
V(f)\gtreqless v
\quad\Longleftrightarrow\quad
\sum_{A\in\T}q_v(A)\Theta_v\!\left(m_{v,A}(f)\right)
\gtreqless \Theta_v(v),
\end{equation}
with the same relation on both sides;
\item \emph{local neutral separation}: for every $v$, distinct cells $A,B$, and states $s\in A$, $t\in B$, there do not exist $r>0$ and $\eta>0$ such that, for every sufficiently small feasible $|\varepsilon|<\eta$ for which the two perturbed cell means remain in $J_v$,
\[
q_v(A)\!\left[\Theta_v\!\left(v+p_{v,A}(s)\varepsilon\right)-\Theta_v(v)\right]
+
q_v(B)\!\left[\Theta_v\!\left(v-p_{v,B}(t)r\varepsilon\right)-\Theta_v(v)\right]
=0.
\]
\end{enumerate}
\end{definition}

The preceding measurement and calibration ingredients can now be collected into the local characterization.

\begin{proposition}\label{prop:local-baseline}
Relative to a fixed admissible elicitation construction, A1--A11 hold if and only if the preference admits a regular neutrally separated local doubly implicit representation in the sense of \Cref{def:local-di}.
\end{proposition}

\paragraph{Proof sketch.}
A1--A4 recover the objective-risk scale, certainty equivalents, revealed cells, and within-cell conditional beliefs. A5 makes the conditional mean sufficient inside $J_v^{\T}$. A6--A9 measure cross-cell compensation on \(I_v\), while A10 decomposes the resulting cell-specific difference scales into a level-\(v\) ambiguity-attitude scale and cell-specific multipliers, whose normalization yields the cell weights \(q_v\). The auxiliary core $K_v\subset I_v$ supplies the small directly measured differences used in the balancing argument; A11(i) imposes background-independent compensation throughout $J_v$ and therefore subsumes A6, while A11(ii) asks whether the full measured scale on $I_v$ has a behaviorally compatible extension throughout $J_v$. Finite pairwise balancing then gives \eqref{eq:local-di}. Local neutral separation follows from the revealed-cell construction. If, for states \(s\) and \(t\) in distinct revealed cells, there existed a rate \(r>0\) for which the numerical equality in the local-neutral-separation clause of \Cref{def:local-di} held for every sufficiently small feasible positive and negative \(\varepsilon\), then the local threshold representation would make the corresponding two-state transfers indifferent to \(\ell_v\). This would imply \(s\sim_v^0 t\), contradicting their membership in distinct revealed cells. Conversely, a regular neutrally separated local DI representation delivers exactly these local restrictions. The appendix gives the two directions in detail. \hfill$\square$

\paragraph{Identification scope.}
The local characterization is relative to the fixed admissible elicitation construction. A11 guarantees at least one behaviorally valid continuation of the measured scale across $J_v$, but it does not by itself establish uniqueness of that continuation where local compensation choices are silent. Construction-independent identification of perception is established after globalization; identification of one common cross-level attitude class requires the later stability restrictions.

\subsection{Global doubly implicit baseline}\label{sec:gbc-globalization}

We now use the tagged-change notation on the full coordinate domain. A tagged change $[a,b]_A$ with $a,b\in[0,1]$ is called a \emph{global tagged change}.

Consider two global tagged changes in distinct cells,
\(
I=[a,b]_A
\;\text{and}\;
J=[c,d]_B.
\)
As in the local construction, the compensation experiment applies \(I\) together with the reverse of \(J\): the \(A\)-coordinate changes from \(b\) to \(a\), while the \(B\)-coordinate changes from \(c\) to \(d\). The required initial coordinates are therefore \(b\) in cell \(A\) and \(c\) in cell \(B\). Define the set of \(v\)-indifferent backgrounds supporting these initial coordinates by

$$
\mathcal Z_v^{A,B}(b,c)
:=
\left\{
z\in[0,1]^{\T}:
z_A=b,\ z_B=c,\ \widehat f_z\sim\ell_v
\right\}.
$$

The transformed block vector $z^{I,-J}$ is the one defined above: it replaces the two prescribed starting coordinates $b,c$ by $a,d$, respectively. For the remainder of this subsection maintain A1--A11.

For global changes \(I=[a,b]_A\) and \(J=[c,d]_B\), call the pair \emph{balance-feasible at level \(v\)} if
$
\mathcal Z_v^{A,B}(b,c)\neq\varnothing.
$
Balance-feasibility requires that the prescribed starting coordinates \(b\) and \(c\) can be jointly embedded in a \(v\)-indifferent block vector. A balance-feasible pair \((I,J)\) is \emph{globally exactly balanced at level \(v\)} if there exists
$
z\in\mathcal Z_v^{A,B}(b,c)
$
such that
$
\widehat f_{z^{I,-J}}\sim\ell_v.
$

\begin{axiom}[Global behavioral calibratability]\label{ax:a12}
For every $v\in(0,1)$ the following full-domain requirements hold.
\begin{enumerate}[label=(\roman*),leftmargin=2.1em]
\item \emph{Global within-cell mean invariance.} If $f_x\sim\ell_v$ and $s,t$ lie in the same revealed cell, then every feasible transfer
\(
f_{x+\varepsilon e_s-r_{st}^v\varepsilon e_t}
\)
is also indifferent to $\ell_v$.
\item \emph{Global background calibration.} For every ordered pair of distinct cells $A,B$, every $a,b,c,d\in[0,1]$, with $I=[a,b]_A$ and $J=[c,d]_B$, all $z,\widetilde z\in\mathcal Z_v^{A,B}(b,c)$, and every $R\in\{\succ,\sim,\prec\}$,
\[
\widehat f_{z^{I,-J}}\,R\,\ell_v
\quad\Longleftrightarrow\quad
\widehat f_{\widetilde z^{I,-J}}\,R\,\ell_v.
\]
\item \emph{Global compatible calibration.} There exists a continuous strictly increasing function $\Phi_v:[0,1]\to\mathbb R$ satisfying $\Phi_v=\phi_v$ on $I_v$ such that equal $\Phi_v$-differences preserve global exact compensation throughout the full domain. More precisely, if $I=[a,b]_A$ and $J=[c,d]_B$ are globally exactly balanced and replacements $I'=[a',b']_A$, $J'=[c',d']_B$ satisfy
\[
\Phi_v(a')-\Phi_v(b')=\Phi_v(a)-\Phi_v(b),
\qquad
\Phi_v(c')-\Phi_v(d')=\Phi_v(c)-\Phi_v(d),
\]
then $(I',J')$ is globally exactly balanced whenever it is balance-feasible at level $v$. One-sided replacement is included by taking $I'=I$ or $J'=J$.
\end{enumerate}
\end{axiom}

A12 is intentionally stronger than the preceding local conditions: A12(i) subsumes A5, A12(ii) subsumes A11(i), and A12(ii)--(iii) jointly subsume A11(ii). These earlier restrictions are retained because they characterize the weaker local model before globalization; the axiom sequence is cumulative rather than logically minimal.

\begin{theorem}\label{thm:global-di-baseline}
Relative to a fixed admissible elicitation construction, the following are equivalent:
\begin{enumerate}[label=(\roman*),leftmargin=2.2em]
\item A1--A12 hold.
\item The preference admits a neutrally separated regular doubly implicit identifiable smooth-ambiguity representation in the sense of \Cref{def:diisa}, with a fixed partition containing at least three cells. Equivalently, there are global level-specific indices $\Phi_v:[0,1]\to\R$ such that, for every $f\in\F$ and $v\in(0,1)$,
\begin{equation}\label{eq:global-di}
V(f)\gtreqless v
\quad\Longleftrightarrow\quad
\sum_{A\in\T}q_v(A)\,\Phi_v\!\left(m_{v,A}(f)\right)
\gtreqless \Phi_v(v),
\end{equation}
with the same relation on both sides.
\end{enumerate}

\end{theorem}

Under the equivalent conditions of \Cref{thm:global-di-baseline}, any global DI representation supplied by the theorem yields, for every act $f$ with $V(f)\in(0,1)$, the doubly implicit certainty-equivalent equation
\begin{equation}\label{eq:global-di-functional}
V(f)
=
\Phi_{V(f)}^{-1}\!\left(
\sum_{A\in\T}q_{V(f)}(A)\,
\Phi_{V(f)}\!\left(
\sum_{s\in A}p_{V(f),A}(s)u(f(s))
\right)
\right).
\end{equation}
Indeed, this follows by setting $v=V(f)$ in \eqref{eq:global-di}. The equation is implicit because the attitude index, conditional beliefs, and cell weights on the right-hand side all depend on the endogenous certainty-equivalent level $V(f)$.

The behaviorally recovered $u$, $\T$, $p_{v,A}$, and $q_v$ enter this global representation directly. A12(iii) requires at least one global index $\Phi_v$ that agrees with the entire measured scale on $I_v$ and whose differences remain coherent with global exact-compensation choices throughout the full domain. Neither A11 nor A12 by itself establishes uniqueness of this behaviorally valid continuation where compensation choices do not distinguish among continuations. Thus \eqref{eq:global-di-functional} is an admissible doubly implicit certainty-equivalent representation, not a claim that the particular global continuation $\Phi_v$ is itself uniquely identified.

\paragraph{Proof sketch.}
A1--A11 give the local DI representation. A12(iii) supplies, for each level $v$, a continuous strictly increasing global continuation $\Phi_v:[0,1]\to\mathbb R$ of the measured scale on $I_v$ and requires equal $\Phi_v$-differences to preserve exact compensation throughout the full domain. A12(ii) makes the resulting global compensation verdict independent of the particular supporting $v$-indifferent background. Together, these two clauses allow any balance-feasible pair of global changes with cancelling weighted $\Phi_v$-differences to be reduced to sufficiently small calibrated changes inside the directly measured region and then transported back to the original pair.

Applying the finite balancing argument to such pairwise moves yields the global threshold representation for all block acts. A12(i) then globalizes within-cell mean sufficiency: any arbitrary act has the same comparison with $\ell_v$ as the block act with the same revealed conditional means, so the block-act representation extends to the full act domain. Finally, since $\Phi_v$ agrees with the directly measured scale on $I_v$, the local neutral-separation property obtained under A1--A11 carries over to the global representation. Conversely, a regular neutrally separated global DI representation immediately implies A1--A11 and the three full-domain clauses of A12. \hfill$\square$

The theorem supplies the global baseline for the cross-level stability refinements. At this stage ambiguity perception \((p_{v,A},q_v)\) may depend on the certainty-equivalent level, and the representation uses a level-specific global index \(\Phi_v\). The latter need not be uniquely determined by choice outside the regions on which compensation behavior disciplines its continuation. Consequently, differences between particular admissible global continuations \(\Phi_v\) and \(\Phi_w\) do not by themselves establish that ambiguity attitude varies behaviorally across levels.

\subsection{Successive stability refinements}\label{sec:postbaseline}

\subsubsection{Stabilizing ambiguity attitude}\label{sec:attitude-stability}

The global doubly implicit baseline leaves open whether ambiguity attitude genuinely varies across certainty-equivalent levels. Although the representation permits a level-specific global index \(\Phi_v\) at each \(v\), such indices may reflect nonunique continuations of the directly measured attitude scales rather than behaviorally meaningful differences in attitude. Attitude stability must therefore be tested using the attitude comparisons directly revealed by choice rather than by comparing arbitrary global continuations \(\Phi_v\).

The first restriction (A13) asks whether the locally revealed attitude-difference orders are compatible across certainty-equivalent levels. If they are, the corresponding directly measured scales can be aligned into one common positive-affine attitude class; the subsequent restriction (A14) asks whether this common class extends to the boundary and governs compensation globally.

\begin{axiom}[Local cross-level attitude consistency]\label{ax:a13}
For every $v,w\in(0,1)$, if $a,b,c,d\in I_v\cap I_w$, then
\[
[a,b]\succeq_v^\Delta[c,d]
\quad\Longleftrightarrow\quad
[a,b]\succeq_w^\Delta[c,d].
\]
\end{axiom}

A13 contains no global attitude function. Its implication for the directly measured scales is recorded next.

\begin{lemma}[Cross-level scale alignment]\label{lem:cross-level-scale-alignment}
Suppose A1--A10 hold. If A13 holds, the directly measured scales $\{\phi_v:I_v\to\mathbb R\}_v$ can be positively affinely normalized so that they agree on overlaps. They therefore define one continuous strictly increasing function $\phi:(0,1)\to\mathbb R$, unique up to one common positive affine transformation, whose restriction to each $I_v$ represents the directly measured objective-difference order.
\end{lemma}

Thus A13 links the level-specific directly measured attitude scales into one common positive-affine class. Fix any representative \(\phi\) of this recovered class and define

$$
\delta^\circ(a,b):=\phi(a)-\phi(b),
\qquad a,b\in(0,1).
$$

The domain is initially \((0,1)^2\) because A13 aligns the directly measured attitude scales only across interior certainty-equivalent levels. It does not by itself ensure that the corresponding attitude differences extend continuously and finitely to comparisons involving \(0\) or \(1\). A14(i) therefore imposes continuous real-valued extension of this interior difference scale to the boundary, while A14(ii) uses the resulting global difference scale to transport exact compensation on the full domain.

These requirements depend only on the positive-scale equivalence class of \(\delta^\circ\). Indeed, replacing \(\phi\) by \(\alpha+\lambda\phi\), with \(\lambda>0\), multiplies every value of \(\delta^\circ\) by the same positive constant. Hence neither the existence of a continuous boundary extension nor the difference equalities used in A14 depend on the particular normalization of \(\phi\).

\begin{axiom}[Global attitude extension and transport]\label{ax:a14}
\begin{enumerate}[label=(\roman*),leftmargin=2.1em]
\item \emph{Boundary regularity}: the recovered common interior difference scale $\delta^\circ$ admits a continuous real-valued extension $\delta:[0,1]^2\to\R$ with $\delta=\delta^\circ$ on $(0,1)^2$.
\item \emph{Common-scale difference transport}: for every $v$, every pair of distinct cells, and every globally exactly balanced pair $I=[a,b]_A$, $J=[c,d]_B$, if replacements $I'=[a',b']_A$, $J'=[c',d']_B$ satisfy
\(
\delta(a',b')=\delta(a,b)
\;\text{and}\;
\delta(c',d')=\delta(c,d),
\)
then $(I',J')$ is globally exactly balanced whenever it is balance-feasible at level $v$. One-sided replacement is the special case $I'=I$ or $J'=J$.
\end{enumerate}
\end{axiom}

A13 is the construction-relative cross-level consistency restriction stated in revealed-preference terms; A14 imposes boundary regularity and global transport on the common cardinal object that A13 has already recovered. In particular, A14(i) is a substantive boundary-extendability condition rather than an assertion that an arbitrary numerical representation exists.
Any extension in A14(i) is necessarily unique. Moreover, continuity extends the interior difference identity
\(\delta(a,c)=\delta(a,b)+\delta(b,c)\)
to all $a,b,c\in[0,1]$.

\begin{theorem}\label{thm:main}
Fix an admissible elicitation construction and suppose the global baseline axioms A1--A12 hold relative to it. The following are equivalent:
\begin{enumerate}[label=(\roman*),leftmargin=2.2em]
\item A13 and A14 hold.
\item The preference admits a neutrally separated regular implicit
identifiable smooth-ambiguity representation with stable attitude in the
sense of \Cref{def:iisa}, with a fixed partition containing at least three
cells. Equivalently, there is one continuous strictly increasing
$\phi:[0,1]\to\R$ such that, for every $f\in\F$ and $v\in(0,1)$,
\begin{equation}
V(f)\gtreqless v
\quad\Longleftrightarrow\quad
\sum_{A\in\T}q_v(A)\,
\phi\!\left(m_{v,A}(f)\right)
\gtreqless \phi(v),
\end{equation}
with the same relation on both sides.
\end{enumerate}

\end{theorem}

Under the equivalent conditions of \Cref{thm:main}, every act $f$ with $V(f)\in(0,1)$ satisfies the stable-attitude implicit certainty-equivalent equation
\begin{equation}\label{eq:stable-attitude-functional}
V(f)
=
\phi^{-1}\!\left(
\sum_{A\in\T}q_{V(f)}(A)\,
\phi\!\left(
\sum_{s\in A}p_{V(f),A}(s)u(f(s))
\right)
\right).
\end{equation}
This follows by setting $v=V(f)$ in the threshold representation of \Cref{thm:main}. Relative to \eqref{eq:global-di-functional}, the ambiguity-attitude index is now fixed, while the conditional beliefs and cell weights may still depend on the endogenous certainty-equivalent level.

\paragraph{Proof sketch.}
By the global doubly implicit baseline, A1--A12 already give, at each certainty-equivalent level $v$, a global representation with level-specific ambiguity-attitude index $\Phi_v$. A13 removes the local part of this level dependence: together with \Cref{lem:cross-level-scale-alignment}, it aligns the directly measured scales across overlapping neighborhoods so that, after positive-affine renormalization, their restrictions are represented by one continuous strictly increasing common scale $\phi:(0,1)\to\mathbb R$. A14(i) extends the associated common difference scale continuously to the boundary and thereby yields a continuous strictly increasing extension of $\phi$ to $[0,1]$.

It remains to show that this common scale, rather than the separate level-specific indices $\Phi_v$, governs compensation behavior globally. A14(ii) provides exactly this transport property: if two changes have the same respective common-scale differences as an already exactly balanced pair, then the replacement pair is also exactly balanced whenever it is balance-feasible. Combining this with the global background calibration supplied by A12(ii), any balance-feasible pair whose $q_v$-weighted $\phi$-differences cancel can be reduced to sufficiently small directly measured changes and then transported back to the original global pair. Hence such cancelling pairwise moves preserve $v$-indifference throughout the full domain.

Applying the finite balancing schedule to the changes from the constant block vector $v\mathbf 1$ to an arbitrary block vector $z$ therefore gives
\(
\widehat f_z\sim\ell_v
\;\Longleftrightarrow\;
\sum_{A\in\mathcal T}q_v(A)\phi(z_A)=\phi(v),
\)
with the strict cases following from continuity and strict state monotonicity. The global within-cell mean sufficiency already supplied by A12(i) then extends this block-act representation to arbitrary acts. Thus the same ambiguity-attitude index $\phi$ is valid at every certainty-equivalent level, while the conditional beliefs $p_{v,A}$ and cell weights $q_v$ may still vary with $v$, yielding the global stable-attitude representation.
Conversely, a representation with one common global $\phi$ immediately makes the measured difference order invariant across levels and satisfies the boundary-extension and common-scale transport requirements in A13--A14. \hfill$\square$

The two main attitude-stability restrictions perform distinct tasks. A13 rules out cross-level disagreement in the directly measured attitude-difference orders on overlaps \(I_v\cap I_w\), thereby aligning the level-specific measured scales into one common positive-affine class. A14 then asks whether this common interior class extends continuously to the boundary and governs compensation on the full domain. The additional intermediate strengthening A13\(^{*}\), studied in \Cref{sec:alternative-branches}, asks whether the A13-common class already governs compensation throughout each larger local representation neighborhood \(J_v\).

The examples in the Appendix separate these requirements. \Cref{ex:a13-failure} satisfies A1--A12 but violates A13, whereas \Cref{ex:a14-transport-failure} satisfies A13\(^{*}\) and A14(i) but violates the global transport condition A14(ii). Thus, validity of one common attitude class throughout every \(J_v\) does not by itself imply that the same class governs compensation on the full domain. A distinct identification issue appears in \Cref{ex:nonunique-global-continuation}: even when a common global ambiguity-attitude class exists, the encompassing global DI representation may admit level-specific indices that are not globally positive-affine equivalent to that class.

Accordingly, A13--A14 identify one common positive-affine ambiguity-attitude class that can be used consistently across certainty-equivalent levels and throughout the global compensation system; they do not require every admissible level-specific index \(\Phi_v\) in the broader DI representation to belong to that class. The remaining structural level dependence in the stable-attitude characterization is therefore confined to ambiguity perception, through the conditional beliefs \(p_{v,A}\) and cell weights \(q_v\). The next axiom asks whether these perception margins are stable across levels as well.

\subsubsection{Cross-level neutral-rate and compensation invariance}\label{sec:stable-perception}

Having stabilized ambiguity attitude, the main characterization path leaves only ambiguity perception potentially dependent on the certainty-equivalent level. The next restriction therefore asks whether the two recovered perception margins are stable across levels: the within-cell conditional beliefs, revealed through neutral-transfer rates, and the across-cell weights, revealed through cross-cell compensation comparisons.

The restriction is formulated directly in terms of these behavioral comparisons and is therefore meaningful independently of the order in which the stability refinements are imposed. On the main characterization path, where A13--A14 have already established a common ambiguity-attitude class, invariance of within-cell neutral rates and of the directly measured cross-cell compensation order has the interpretation of perception stability. Without A13--A14, however, the same compensation-invariance requirement can also restrict the level-specific attitude scales, as shown in \Cref{prop:perception-local}.

By A4, every pair of distinct states in the same revealed cell has a neutral transfer rate at every level, and A3 makes that rate unique.

\begin{axiom}[Cross-level neutral-rate and compensation invariance]\label{ax:a15}
For every pair of certainty-equivalent levels $v,w\in(0,1)$:
\begin{enumerate}[label=(\roman*),leftmargin=2.2em]
\item \emph{Within-cell neutral-transfer stability.} For every pair of distinct states $s,t$ in the same revealed cell,
\[
r_{st}^v=r_{st}^w.
\]

\item \emph{Cross-cell compensation stability.} Whenever $I_v\cap I_w\neq\varnothing$, for every pair of distinct revealed cells $A,B\in\T$ and every $a,b,c,d\in I_v\cap I_w$,
\[
[a,b]_A\succeq_v^F[c,d]_B
\;\Longleftrightarrow\;
[a,b]_A\succeq_w^F[c,d]_B.
\]
\end{enumerate}
\end{axiom}

Under the stable-attitude representation supplied by A13--A14, A15(i) fixes the conditional beliefs $p_{v,A}$ across levels, while A15(ii) fixes the cell-weight distribution $q_v$. Thus, on the main characterization path, A15 removes the remaining utility-level dependence in perception; outside that path its formal content is the neutral-rate and compensation invariance just stated.

\begin{definition}\label{def:standard-sa}
A preference admits a \emph{regular standard identifiable smooth-ambiguity representation with fixed perception} if it admits a regular stable-attitude representation in the sense of \Cref{def:iisa} for which there exist conditional beliefs $\{p_A\in\operatorname{int}\Delta(A)\}_{A\in\T}$ and a weight vector $q\in\operatorname{int}\Delta(\T)$ satisfying
\[
p_{v,A}=p_A
\quad\text{and}\quad
q_v=q
\qquad\text{for every }v\in(0,1),\ A\in\T.
\]
Equivalently, its threshold representation uses one fixed partition, one fixed collection of within-cell beliefs, one fixed across-cell weight vector, and one common ambiguity-attitude index at every certainty-equivalent level. If the underlying regular representation is neutrally separated in the sense of \Cref{def:neutral-separation}, we call the standard representation neutrally separated.
\end{definition}

\begin{theorem}\label{thm:standard-sa}
Fix an admissible elicitation construction and suppose A1--A14 hold relative to it. The following are equivalent:
\begin{enumerate}[label=(\roman*),leftmargin=2.2em]
\item A15 holds.
\item The preference admits a neutrally separated standard identifiable
smooth-ambiguity representation with fixed perception in the sense of \Cref{def:standard-sa}: there exist a fixed
partition $\T$ containing at least three cells, fixed conditional beliefs
$\{p_A\in\operatorname{int}\Delta(A)\}_{A\in\T}$, fixed weights
$q\in\operatorname{int}\Delta(\T)$, the normalized vNM index $u$, and one
continuous strictly increasing $\phi:[0,1]\to\mathbb R$ such that, for every
$f\in\F$ and $v\in(0,1)$,
\begin{equation}
V(f)\gtreqless v
\quad\Longleftrightarrow\quad
\sum_{A\in\T}q(A)\,
\phi\!\left(
\sum_{s\in A}p_A(s)u(f(s))
\right)
\gtreqless \phi(v),
\end{equation}
with the same relation on both sides.
\end{enumerate}

\end{theorem}

Under the equivalent conditions of \Cref{thm:standard-sa}, the threshold representation yields the standard smooth-ambiguity certainty-equivalent functional form
\begin{equation}\label{eq:standard-sa-functional}
V(f)
=
\phi^{-1}\!\left(
\sum_{A\in\T}q(A)\,
\phi\!\left(
\sum_{s\in A}p_A(s)u(f(s))
\right)
\right),
\qquad f\in\F.
\end{equation}
Unlike \eqref{eq:global-di-functional} and \eqref{eq:stable-attitude-functional}, the right-hand side of \eqref{eq:standard-sa-functional} contains no object indexed by $V(f)$: ambiguity attitude and both perception margins are fixed.

Thus, relative to the fixed admissible elicitation construction and within the maintained environment, A1--A15 characterize the standard identifiable smooth-ambiguity benchmark. These background restrictions serve identification and are not part of the general definition of smooth ambiguity.

\paragraph{Proof sketch.}
By \Cref{thm:main}, A1--A14 already yield the global stable-attitude representation with one common ambiguity-attitude index \(\phi:[0,1]\to\mathbb R\); the only remaining level dependence is therefore in the perception terms \(p_{v,A}\) and \(q_v\).
A15(i) removes the within-cell component of that dependence. Since the revealed neutral-transfer rate between distinct states \(s,t\in A\) is
\(
r_{st}^v=\frac{p_{v,A}(s)}{p_{v,A}(t)},
\)
its invariance across certainty-equivalent levels makes all pairwise probability ratios within a cell level independent. Full-support normalization then gives one common conditional belief \(p_A\) for every nonsingleton cell, while the conclusion is immediate for singleton cells. It remains to stabilize the across-cell weights. Fix two levels \(v,w\) with \(I_v\cap I_w\neq\varnothing\). Because the ambiguity-attitude scale \(\phi\) is already common, a directly measured cross-cell indifference at level \(v\) has the form
\(
q_v(A)\,[\phi(a)-\phi(b)]
=
q_v(B)\,[\phi(c)-\phi(d)].
\)
A15(ii) requires the same cross-cell compensation comparison at level \(w\). Hence the same two nonzero common-scale differences must satisfy
\(
q_w(A)\,[\phi(a)-\phi(b)]
=
q_w(B)\,[\phi(c)-\phi(d)],
\)
which implies
\(
\frac{q_v(A)}{q_v(B)}
=
\frac{q_w(A)}{q_w(B)}.
\)
Thus every pairwise weight ratio is constant across overlapping measurement neighborhoods. Since the intervals \(\{I_v\}_{v\in(0,1)}\) cover the connected interval \((0,1)\), these equalities propagate along finite chains of overlaps. Normalization then yields one common weight vector \(q\).
Conversely, under such a fixed-perception representation the neutral-transfer rates \(p_A(s)/p_A(t)\) and the cross-cell compensation order
\(
q(A)[\phi(a)-\phi(b)]
\gtreqless
q(B)[\phi(c)-\phi(d)]
\)
contain no level-specific object. Therefore both clauses of A15 hold. \hfill$\square$

\paragraph{Main characterization path.}
Within this environment, the primary results form a nested sequence. The restrictions are intentionally cumulative rather than claimed to be logically independent; earlier clauses isolate the local structure needed for the intermediate characterizations, while later full-domain clauses may subsume those local restrictions. Table~\ref{tab:main-characterization} summarizes this principal sequential characterization path, and its representation labels are understood within that environment throughout.

\begin{table}[ht]
\centering
\footnotesize
\renewcommand{\arraystretch}{1.16}
\begin{tabular}{@{}>{\raggedright\arraybackslash}p{0.22\textwidth}>{\raggedright\arraybackslash}p{0.23\textwidth}>{\raggedright\arraybackslash}p{0.20\textwidth}>{\raggedright\arraybackslash}p{0.25\textwidth}@{}}
\toprule
\textbf{Behavioral restrictions} & \textbf{Ambiguity attitude} & \textbf{Ambiguity perception} & \textbf{Characterized representation} \\
\midrule
\textbf{A1--A12} & Level-specific DI indices permitted & May vary with $v$ & Global doubly implicit SA \\
\textbf{+ A13--A14} & Common globally & May vary with $v$ & Global stable-attitude implicit SA \\
\textbf{+ A15} & Common globally & Fixed & Standard identifiable SA \\
\bottomrule
\end{tabular}
\caption{Main characterization path through the global representation classes.}
\label{tab:main-characterization}
\end{table}

\Cref{thm:standard-sa} closes the main characterization path: A13--A14 establish one common global ambiguity-attitude class, and A15 removes the remaining structural level dependence in ambiguity perception. These characterization results are stated relative to a fixed admissible elicitation construction. The next subsection asks which of the recovered representation objects are invariant across such constructions and hence behaviorally identified.

\subsection{Identification}\label{sec:identification}

A1--A11 characterize the local doubly implicit representation relative to a chosen admissible elicitation construction, and \Cref{thm:global-di-baseline} shows that A12 globalizes that representation. At the local stage we use \emph{recover} for objects obtained relative to the fixed neighborhoods $\{(I_v,J_v)\}_{v\in(0,1)}$ and orientation order $\triangleleft$, reserving \emph{identify} for objects shown to be invariant to those auxiliary choices. The identification result has two layers. Under A1--A12, globalization already permits construction-independent identification of ambiguity perception at each fixed certainty-equivalent level: the partition and within-cell beliefs come directly from the primitive neutral-transfer relation, while the fact that any two global DI representations of the same preference generate the same local \(v\)-indifference set pins down the normalized across-cell weights. The attitude side has two distinct sources of representational nonuniqueness. First, the directly measured scale on $I_v$ is cardinal and therefore retains the usual level-specific positive-affine normalization freedom. Second, even after fixing a normalization on $I_v$, A12(iii) may admit several behaviorally compatible full-domain continuations $\Phi_v$ that are not positive-affine transformations of one another. Thus the baseline need not identify a unique full-domain level-specific attitude class. A13--A14 do not eliminate every such DI continuation; rather, they link the level-specific measured attitude classes and identify the unique common positive-affine class that can be used consistently across all certainty-equivalent levels and extended to govern global compensation behavior.

\begin{lemma}\label{lem:family-invariance}
Suppose the same preference satisfies A1--A12 relative to two admissible
elicitation constructions, possibly with different elicitation neighborhoods
and different strict total orders on the revealed cells. Let
\(
(p_{v,A},q_v,\Phi_v)
\;\text{and}\;
(\widetilde p_{v,A},\widetilde q_v,\widetilde\Phi_v)
\)
be global doubly implicit representation components supplied by
\Cref{thm:global-di-baseline} under the two constructions, choosing the
positive-affine representative of each global index so that
\(
\Phi_v|_{I_v}=\phi_v
\;\text{and}\;
\widetilde\Phi_v|_{\widetilde I_v}=\widetilde\phi_v,
\)
where $\phi_v$ and $\widetilde\phi_v$ are the directly measured level-$v$
scales under the respective constructions.

\begin{enumerate}[label=(\roman*),leftmargin=2.2em]

\item \emph{Fixed-level identification.}
For every $v$ and $A$,
\(
\widetilde p_{v,A}=p_{v,A}
\;\text{and}\;
\widetilde q_v(A)=q_v(A).
\)
Moreover, for every $v$ there exist an open neighborhood
$N_v\subset I_v\cap\widetilde I_v$ of $v$ and $\lambda_v>0$ such that
\[
\widetilde\Phi_v(a)-\widetilde\Phi_v(v)
=
\lambda_v
\bigl[\Phi_v(a)-\Phi_v(v)\bigr]
\quad
\text{for every }a\in N_v.
\]

\item \emph{Stable-attitude identification.}
If, in addition, A13--A14 hold relative to both constructions, let
$\phi$ and $\widetilde\phi$ be any two common global stable-attitude indices
that represent the same preference under the respective construction-specific
perception components. Then there exist
$\alpha\in\mathbb R$ and $\lambda>0$ such that
\[
\widetilde\phi(a)
=
\alpha+\lambda\phi(a)
\quad
\text{for every }a\in[0,1].
\]

\end{enumerate}
\end{lemma}

Part (i) identifies fixed-level perception and the locally revealed attitude class in a neighborhood of each $v$; part (ii) identifies the common global stable-attitude class up to one common positive affine transformation.

The attitude-identification conclusion can be summarized schematically by separating the larger set of level-specific DI representations from the common structural class selected by the stable-attitude restrictions. Relative to the fixed admissible elicitation construction under discussion, for each certainty-equivalent level $v$, let $\mathcal A_v$ be the set of all continuous strictly increasing functions $\Psi:[0,1]\to\mathbb R$ such that, holding the behaviorally identified $p_{v,A}$ and $q_v$ fixed, (i) the level-$v$ threshold relation is represented by
\[
V(f)\gtreqless v
\quad\Longleftrightarrow\quad
\sum_{A\in\T}q_v(A)\Psi\!\left(m_{v,A}(f)\right)
\gtreqless \Psi(v),
\]
(ii) $\Psi|_{I_v}$ belongs to the positive-affine class of the directly measured scale $\phi_v$, and (iii) equal $\Psi$-differences preserve global exact compensation at level $v$ as in A12(iii). By construction, $\mathcal A_v$ is closed under positive affine transformations.

The global DI baseline may leave $\mathcal A_v$ containing several distinct positive-affine classes. A13--A14 do not collapse $\mathcal A_v$ to one class. Instead, they identify the unique common positive-affine class $[\phi]$ satisfying
\[
[\phi]\subseteq\mathcal A_v
\qquad\text{for every }v\in(0,1),
\]
so that the same class is behaviorally operative across levels and on the full domain. To see the uniqueness explicitly, suppose another positive-affine class $[\psi]$ were contained in every $\mathcal A_v$. A single representative $\psi$ would then supply the same preference with another common global stable-attitude representation. Applying part~(ii) of \Cref{lem:family-invariance} with the same admissible elicitation construction on both sides gives $[\psi]=[\phi]$. Thus what A13--A14 identify is a common structural attitude class, not a unique parametrization of the larger global DI representation. In particular, apparent level dependence of the indices $\Phi_v$ in a particular DI representation need not reflect genuine behavioral variation in ambiguity attitude. \Cref{ex:nonunique-global-continuation} in the Appendix shows that these inclusions can be strict: the same preference can admit the common stable-attitude index $\phi(x)=x^2$ and also a global DI representation whose level-specific indices are not all positive-affine transformations of one another.

Lemma~\ref{lem:family-invariance} establishes invariance to the auxiliary elicitation construction. Combining that result with the earlier behavioral recovery of the revealed partition and within-cell neutral rates yields the following identification statement for the economic objects of the model.

\begin{corollary}\label{cor:identification}
Within the neutrally separated, full-support global stable-attitude class characterized by \Cref{thm:main}, the partition is behaviorally identified as
$
\T=S/\!\sim^0.
$
For any reference state \(s_A\in A\), the conditional belief \(p_{v,A}\) is recovered from the neutral-transfer rates by \eqref{eq:p-from-rates}, while part~(i) of \Cref{lem:family-invariance} identifies the normalized cell weights \(q_v\). Hence the predictive probability \(\pi_v\), and equivalently the statistically identifiable perception pair \((P_v,\mu_v)\), are behaviorally identified at each certainty-equivalent level, up to relabelling of cells (equivalently, models). Under A13--A14, part~(ii) of \Cref{lem:family-invariance} additionally identifies the common global ambiguity-attitude class \([\phi]\) up to one common positive affine transformation. This does not imply uniqueness of the larger family of admissible level-specific DI indices.
\end{corollary}

Under A1--A12, perception is identified separately at each certainty-equivalent level but may vary with \(v\). A13--A14 add identification of one common global ambiguity-attitude class, while A15 further requires the already identified perception objects to be invariant across levels. The local representation, the intermediate notion of local attitude stability, and alternative orders of refinement are studied in \Cref{sec:alternative-branches}.

\section{Alternative refinements and extensions}\label{sec:extensions}

This section studies alternative refinement paths relative to the main characterization, isolates an intermediate notion of local attitude stability, and then considers extensions beyond statistically identifiable smooth ambiguity. Throughout this section, ``local DI,'' ``global DI,'' and ``stable-attitude implicit SA'' retain the regularity, full-support, neutral-separation, fixed-partition, and at-least-three-cell qualifications of the corresponding characterization results.

\subsection{Local representation and alternative orders of refinement}\label{sec:alternative-branches}

A1--A11 already characterize the local DI model in \Cref{prop:local-baseline}; A12 is needed only to extend that representation to the full objective-utility domain. This makes it possible to separate local attitude stability from global attitude stability and to consider perception stability before globalization.

\begin{proposition}\label{prop:gbc-a12}
Suppose A1--A12 hold. Then A13 holds if and only if there exists a global DI representation supplied by \Cref{thm:global-di-baseline} whose level-specific indices $\Phi_v$ can be normalized so that their restrictions to the directly measured neighborhoods $I_v$ are generated by one common continuous strictly increasing scale $\phi:(0,1)\to\mathbb R$. The $\Phi_v$ may still differ away from those directly measured neighborhoods.
\end{proposition}

This result isolates the distinction between the common ambiguity-attitude class supplied by A13 on the directly measured neighborhoods and the stronger A14 requirement that one common attitude index govern all global compensation choices.

A second alternative is to impose A15 before global attitude stability. In that order A15 cannot yet be interpreted as a pure perception restriction, because its compensation-invariance clause also aligns the directly measured attitude scales.

\begin{proposition}\label{prop:perception-local}
Suppose A1--A11 hold. Then A15 holds if and only if:
\begin{enumerate}[label=(\roman*),leftmargin=2.2em]
\item the within-cell beliefs and across-cell weights are fixed across levels, so $p_{v,A}=p_A$ for every $v,A$ and $q_v=q$ for every $v$; and
\item A13 holds, so the directly measured ambiguity-attitude scales belong to one common positive-affine class on overlapping measurement neighborhoods.
\end{enumerate}
\end{proposition}

The second implication is important. A15(ii) requires more than constancy of the recovered cell weights: it requires every directly measured cross-cell utility-difference comparison to repeat across levels. That stronger invariance also implies A13 and hence aligns the directly measured ambiguity-attitude scales.

Adding A12 yields the global counterpart of Proposition \(\ref{prop:perception-local}\): a fixed-perception global DI representation whose directly measured attitude scales are aligned across levels.

\begin{corollary}\label{cor:perception-gbc}
Suppose A1--A11 hold. Then A12 and A15 are equivalent to a global DI representation with fixed $\{p_A\}_{A\in\T}$ and $q$ whose level-specific global attitude indices can be normalized to agree with one common attitude scale on the directly measured neighborhoods $I_v$. The global indices may still differ away from those neighborhoods unless A14 is imposed in addition.
\end{corollary}

Maintaining A1--A12 throughout, Table~\ref{tab:global-classes} collects the global representation classes generated by alternative combinations of the global attitude-stability restrictions A13--A14 and the perception-stability restriction A15. Because these restrictions are not behaviorally independent, the table records alternative restriction paths rather than a partition of preferences into mutually exclusive classes.

\begin{table}[ht]
\centering
\small
\begin{tabular}{@{}>{\raggedright\arraybackslash}p{0.21\textwidth}>{\raggedright\arraybackslash}p{0.34\textwidth}>{\raggedright\arraybackslash}p{0.34\textwidth}@{}}
\toprule
&
\textbf{A15 not imposed}
&
\textbf{A15 imposed}
\\
\midrule
\textbf{A13--A14 not imposed}
&
Global DI (\Cref{thm:global-di-baseline})
&
Fixed-perception global DI; measured scales aligned (\Cref{cor:perception-gbc})
\\[1ex]
\textbf{A13--A14 imposed}
&
Global stable-attitude model (\Cref{thm:main})
&
Standard identifiable smooth ambiguity (\Cref{thm:standard-sa})
\\
\bottomrule
\end{tabular}
\caption{Global representation classes generated by alternative combinations of stability restrictions.}
\label{tab:global-classes}
\end{table}

The asymmetry is clearest in the two off-diagonal classes. By \Cref{prop:perception-local}, imposing perception stability through A15 also entails the cross-level attitude consistency in A13; by contrast, imposing global attitude stability through A13--A14 does not analogously restrict ambiguity perception. Consequently, \Cref{cor:perception-gbc} characterizes a narrower class than the purely representation-theoretic case in which perception is fixed while the level-specific attitude indices are otherwise unrestricted.

\subsection{Local stable attitude}\label{sec:local-stable-attitude}

Table~\ref{tab:global-classes} concerns global representation classes obtained once A12 is maintained. A separate intermediate question is whether the common attitude class aligned by A13 already governs compensation throughout the larger local representation neighborhoods \(J_v\), without imposing the full global transport requirement A14. To study this question, return to the A1--A11 local environment and strengthen A13 as follows.

\begin{conditionstar}[Local common-scale consistency]
A13\(^{*}\) consists of the following two sequential requirements:
\begin{enumerate}[label=(\roman*),leftmargin=2.2em]

\item \emph{Cross-level local attitude consistency.}
For every \(v,w\in(0,1)\), if
\(a,b,c,d\in I_v\cap I_w\), then

$$
[a,b]\succeq_v^\Delta[c,d]
\quad\Longleftrightarrow\quad
[a,b]\succeq_w^\Delta[c,d].
$$

Thus A13 holds. Let \([\phi]\) denote the common positive-affine attitude class generated by this requirement through \Cref{lem:cross-level-scale-alignment}. For each \(v\), given the chosen normalized directly measured representative \(\phi_v\), let \(\phi^v\) denote the unique representative of \([\phi]\) whose restriction to \(I_v\) equals \(\phi_v\).

\item \emph{Local common-scale transport.}
For every $v$, every pair of distinct revealed cells $A,B$, every locally
exactly balanced pair
\(
I=[a,b]_A
\;\text{and}\;
J=[c,d]_B
\)
with endpoints in $J_v$, and every locally supportable replacement
\(
I'=[a',b']_A
\;\text{and}\;
J'=[c',d']_B
\)
in the same two cells, if
\(
\phi^v(a')-\phi^v(b')
=
\phi^v(a)-\phi^v(b)
\)
and
\(
\phi^v(c')-\phi^v(d')
=
\phi^v(c)-\phi^v(d),
\)
then $(I',J')$ is also locally exactly balanced at level $v$.

\end{enumerate}
\end{conditionstar}

Part~(i) is A13: it aligns the directly measured attitude-difference orders across certainty-equivalent levels and thereby generates the common class \([\phi]\). Part~(ii) requires that, after this class is expressed in level-\(v\) units, its cardinal differences continue to organize exact compensation throughout the larger local representation neighborhood \(J_v\). The representative \(\phi^v\) is unique conditional on the chosen directly measured representative \(\phi_v\); the normalization \(\phi_v(v)=0\) alone still leaves a positive multiplicative freedom.

\begin{proposition}\label{prop:local-stable}
Suppose A1--A11 hold. The following are equivalent:
\begin{enumerate}[label=(\roman*),leftmargin=2.2em]
\item A13$^*$ holds.
\item There is one continuous strictly increasing $\phi:(0,1)\to\mathbb R$, unique up to one common positive affine transformation, such that for every $v$ and every act whose vector of cell means lies in $J_v^{\T}$,
\begin{equation}\label{eq:local-stable}
V(f)\gtreqless v
\quad\Longleftrightarrow\quad
\sum_{A\in\T}q_v(A)\phi\!\left(m_{v,A}(f)\right)
\gtreqless\phi(v).
\end{equation}
\end{enumerate}
\end{proposition}

Thus A13$^*$ is exactly the local stable-attitude condition: A13 aligns the common attitude class on the directly measured neighborhoods \(I_v\), while A13$^*$(ii) requires that this same class govern the full local representation domains \(J_v\). Since \Cref{prop:perception-local} shows that A15 already implies A13, when A15 and A13$^*$ are imposed together the additional content of A13$^*$ is precisely this local common-scale transport requirement.

To obtain one common attitude index on the larger local representation neighborhoods \(J_v\), it remains to impose the local common-scale transport requirement A13\(^*\)(ii).

\begin{corollary}\label{cor:perception-local-common}
Under A1--A11, A15 together with A13$^*$ is equivalent to a local DI representation with fixed $\{p_A\}_{A\in\T}$, fixed $q$, and one common ambiguity-attitude index $\phi$ that represents the level-\(v\) threshold relation on the entire local domain \(J_v^{\mathcal T}\) for every \(v\).
\end{corollary}

The restrictions play distinct but non-independent roles. A12 globalizes the level-specific DI representation, while A15 stabilizes ambiguity perception and already implies the cross-level attitude consistency in A13. Hence, when A15 is combined with A13\(^{*}\), the additional content of A13\(^{*}\) is its local common-scale transport requirement, which makes the A13-common attitude class govern the entire local representation domains \(J_v^{\mathcal T}\).

\subsection{Beyond identifiable smooth ambiguity}
\label{sec:beyond-identifiable}

The decomposition developed above has a representation-level analogue beyond statistically identifiable model families. In a general smooth-ambiguity specification, ambiguity perception at certainty-equivalent level \(v\) may be represented by a second-order belief \(\mu_v\) over probabilistic models \(p\in\Delta(S)\), while ambiguity attitude is represented by an increasing index \(\Phi_v\). The corresponding doubly implicit threshold representation takes the form

\[
V(f)\gtreqless v
\quad\Longleftrightarrow\quad
\int_{\Delta(S)}
\Phi_v\!\left(\E_p[u(f)]\right)\,d\mu_v(p)
\gtreqless
\Phi_v(v).
\]

At this level of generality, the same conceptual hierarchy remains available. Stability of ambiguity attitude asks whether the family \(\{\Phi_v\}_v\) can be replaced by one common index \(\phi\), while stability of ambiguity perception asks whether the family \(\{\mu_v\}_v\) can be replaced by one common second-order belief \(\mu\). Thus the representation-side sequence
\(
(\mu_v,\Phi_v)
\longrightarrow
(\mu_v,\phi)
\longrightarrow
(\mu,\phi)
\)
parallels the main characterization path of the paper.

What does not extend directly is the present behavioral recovery argument. Statistical-model identifiability is what converts the finite-state perception pair into a partition of states, conditional beliefs within cells, and weights across cells. This structure permits the neutral-transfer experiments to recover within-cell probability ratios and the cross-cell compensation experiments to recover cell weights and ambiguity attitude separately. With unrestricted smooth ambiguity, distinct probabilistic models may assign positive probability to the same states, so there need not be a partition whose coordinates can be perturbed independently in this way. A statewise transfer may then affect the expectations of several represented models simultaneously. The neutral-transfer and cross-cell-compensation constructions used above therefore cannot in general be carried over merely by replacing the identifiable perception pair with an unrestricted second-order belief.

This distinction separates a representation question from an identification question. At the representation level, one can define progressively more stable smooth-ambiguity classes according to whether both \((\mu_v,\Phi_v)\) vary with \(v\), the attitude index is common across levels, or the second-order belief is common as well. But the availability of such representations does not imply that observed choice separately reveals whether variation comes from ambiguity perception or ambiguity attitude. The stronger behavioral question is whether these two components can themselves be separately recovered from choice and their cross-level stability tested independently. Outside the identifiable class, such a behavioral decomposition requires additional identifying structure or a different elicitation argument. In particular, behavior need not distinguish different second-order beliefs: when ambiguity attitude is affine, smooth aggregation depends only on the predictive mean of the second-order belief, so different \(\mu_v\)'s with the same mean are observationally equivalent.

The identifiable model studied here should therefore be viewed as a domain on which the distinction between changing perception and changing attitude is not only representational but behaviorally operational. Extending that decomposition to unrestricted smooth ambiguity raises a separate identification problem: under what conditions can changes in a second-order belief over overlapping probabilistic models be distinguished from changes in ambiguity attitude? Resolving that question would require new behavioral structure beyond the partition-based neutral-transfer construction used here.

\section{Related literature}\label{sec:literature}

The paper first relates to work separating ambiguity from attitudes toward ambiguity and identifying the objects that enter smooth-ambiguity representations. \citet{GhirardatoMarinacci2002} develop comparative notions of ambiguity and ambiguity aversion, while \citet{GhirardatoMaccheroniMarinacci2004} separate an unambiguous preference relation from attitudes toward ambiguity. \citet{KMM2005} make the distinction explicit in smooth ambiguity through a vNM index, a separate ambiguity-attitude transformation, and second-order beliefs over probabilistic models. \citet{DentiPomatto2022} then provide an axiomatic foundation for smooth ambiguity with identifiable subjective statistical models and show how model and predictive uncertainty can be recovered from preferences. The present paper adds a different identification problem to this line of work. When ambiguity-sensitive behavior changes across endogenous certainty-equivalent levels, both perception and the level-specific attitude representation may change. In finite states, identifiable perception is represented by within-cell conditional beliefs and across-cell weights; the behavioral construction recovers those margins alongside a local attitude scale, and the cross-level restrictions determine which components are stable. The paper therefore combines two questions that this literature has largely treated separately: decomposition of ambiguity perception from ambiguity attitude, and endogenous variation of that decomposition across utility levels.

A second connection is to empirical and theoretical work on wealth, stakes, and context dependence in ambiguity-sensitive choice. \citet{BaillonPlacido2019} reject constant absolute and relative ambiguity aversion and report evidence in the direction of decreasing ambiguity aversion as individuals become better off. \citet{BouchouichaEtAl2017} document stake effects, while \citet{BaillonBleichrodt2015} find variation across gain--loss and likelihood domains; \citet{TrautmannWakker2018} and \citet{KonigKerstingKopsTrautmann2023} document failures of restrictions retained by prominent ambiguity models. On the theoretical side, \citet{CherbonnierGollier2015} derive conditions for decreasing aversion under ambiguity, including smooth ambiguity, while \citet{CerreiaVioglioMaccheroniMarinacci2022} give behavioral characterizations of increasing or decreasing ambiguity aversion as wealth changes. These contributions ask whether ambiguity-sensitive behavior or ambiguity aversion changes with wealth, stakes, domains, or related environments. The present paper asks a prior identification question: when ambiguity-sensitive behavior varies, does the variation reflect a change in ambiguity attitude, ambiguity perception, or both? The evidence cited above therefore motivates allowing ambiguity-sensitive behavior to vary, but it is not direct evidence of dependence on the certainty-equivalent level. In particular, the certainty-equivalent index \(v\) is not an exogenous wealth or stake variable; it is endogenously determined by the act being evaluated. Exogenous changes in wealth or common certain components may nevertheless shift the relevant certainty-equivalent level and thereby generate comparative statics through \(v\).

The implicit nature of the representation connects the paper to models in which the evaluation supporting an indifference class may depend on the value assigned to the alternative. \citet{Dekel1986} provides an early example by weakening independence and allowing the supporting evaluation to vary with utility. Related work under uncertainty, including \citet{GrantKajiiPolak2000}, \citet{GrantRoordaYang2025}, and \citet{RoordaJoosten2026}, likewise permits local evaluations to depend on the act's represented value or balancing level; in particular, \citet{GrantRoordaYang2025} include specifications in which ambiguity attitude may vary with the balancing value. The present paper introduces such endogenous level dependence within identifiable smooth ambiguity and separates two channels through which it can operate: the act's certainty-equivalent level may select both ambiguity perception and the level-specific attitude representation. The characterization then determines when either source of variation can be removed, rather than treating level dependence as variation in one undifferentiated supporting evaluation.

Finally, \citet{Ozbek2026ULD} develops utility-level-dependent ambiguity in reduced form by allowing a capacity $\nu_v$ to depend directly on the certainty-equivalent level. That model therefore permits the ambiguity object itself to vary with utility but does not decompose the source of that variation. The present paper replaces the reduced-form map $v\mapsto\nu_v$ by the structural decomposition into within-cell beliefs, across-cell weights, and an ambiguity-attitude class, and asks which of these components must vary with $v$. In this sense, the two approaches address complementary questions: the reduced-form model identifies utility-level dependence in ambiguity evaluation, whereas the present smooth-ambiguity framework decomposes such dependence into changing perception, changing attitude, or both.

\section{Conclusion}\label{sec:conclusion}

This paper develops a behavioral theory of implicit smooth ambiguity in which ambiguity perception may depend on the act's endogenous certainty-equivalent level and the global doubly implicit representation may use level-specific ambiguity-attitude indices. The characterization separates direct measurement, local representation, and globalization. Neutral transfers recover the revealed partition and within-cell conditional beliefs; cross-cell compensation recovers cell weights and the directly measured ambiguity-attitude scale; and local calibration determines whether that measured scale can represent choice throughout the larger local neighborhoods. Full-domain calibration then yields the global doubly implicit baseline without requiring the particular level-specific global attitude indices to be uniquely identified.

The cross-level restrictions distinguish progressively stronger forms of stability. A13 asks whether the attitude-difference orders directly revealed at different certainty-equivalent levels can be aligned into one common positive-affine class. A14 asks whether that common class extends to the boundary and governs compensation on the full domain. Conditional on this common global ambiguity-attitude class, A15 tests whether the two components of ambiguity perception, the within-cell conditional beliefs and the across-cell weights, are also invariant across levels. When all of these restrictions hold, the model reduces to standard identifiable smooth ambiguity within the maintained full-support, neutrally separated fixed-partition environment. The alternative characterization paths show that local attitude stability, global attitude stability, and perception stability can also be studied separately and in different orders.

The decomposition has a direct identification interpretation. Perception and ambiguity attitude are recovered from different choice comparisons, and their stability is therefore separately testable. After globalization, ambiguity perception is behaviorally identified at each certainty-equivalent level independently of the auxiliary elicitation construction. The attitude side is more subtle: level-specific global DI indices may remain nonunique, but the cross-level attitude restrictions identify the unique common positive-affine ambiguity-attitude class that can govern compensation consistently across levels and on the full domain. Apparent variation in a particular family of global indices therefore need not itself constitute behavioral variation in ambiguity attitude.

This distinction also gives the model an empirical interpretation. Neutral-transfer experiments recover within-cell probability ratios, while cross-cell compensation experiments recover across-cell weights and local attitude differences. Repeating these measurements at different certainty-equivalent levels makes it possible to test separately whether ambiguity attitude and ambiguity perception remain stable. Observed utility-level dependence in ambiguity-sensitive choice therefore need not automatically be interpreted as changing ambiguity aversion: it may instead arise from changes in conditional beliefs, changes in the relative weights assigned to ambiguity cells, or genuine changes in the locally revealed attitude scale. In applications, this distinction can help organize evidence on how ambiguity-sensitive behavior changes with wealth, stakes, or other economic circumstances that shift the relevant certainty-equivalent level.

Several extensions remain. The two-cell case requires a different identification argument because the present cross-cell measurement uses a third balancing cell. Allowing the revealed partition itself to vary with the certainty-equivalent level would require a theory of how ambiguity cells appear, disappear, split, or merge across levels. Beyond statistically identifiable smooth ambiguity, the representation-level distinction between level-dependent perception and attitude remains meaningful, but new behavioral tools would be needed to recover more general second-order beliefs. Dynamic extensions would additionally require consistency conditions linking perception and attitude across information histories. These directions would clarify how far the behavioral decomposition developed here extends beyond the finite identifiable environment.

\appendix

\section{Preliminary identification and measurement results}\label[appendix]{app:prelim}

\subsection{Finite-state identifiability}\label[appendix]{app:finite-ident}
This subsection proves the finite-state reduction from identifiable model families to a partition and predictive probability used in the behavioral analysis.

\begin{proof}[\textbf{Proof of \Cref{lem:finite-ident}}]
Because $A_p=k^{-1}(p)$, the sets $A_p$ are pairwise disjoint, their union is $S$, and the identifying condition gives $p(A_p)=1$, so every $A_p$ is nonempty. Hence $p'(A_p)=0$ whenever $p'\neq p$. Therefore $\pi(A_p)=\sum_{p'}\mu(p')p'(A_p)=\mu(p)$. Full support of $\pi$ and nonemptiness of $A_p$ imply $\pi(A_p)>0$, hence $\mu(p)>0$. If $s\in A_p$, then only $p$ assigns positive mass to $s$, so $\pi(s)=\mu(p)p(s)$ and therefore $p(s)=\pi(s)/\pi(A_p)$. This gives $p=\pi(\cdot\mid A_p)$.  Conversely, if $k(s)=p_A$ for $s\in A$, then $p_A(\{s:k(s)=p_A\})=p_A(A)=1$.
\end{proof}

\subsection{Verification of the certainty-equivalent construction}\label[appendix]{app:ce-index}
This subsection establishes the certainty-equivalent index used throughout the behavioral construction. Recall that, under A1–A3, every act \(f\) is indifferent to the canonical act \(f_{u(f)}\) and has a unique canonical certainty equivalent \(\ell_{V(f)}\). The resulting index \(V:\mathcal F\to[0,1]\) is continuous.

\begin{proof}
By statewise substitution, replacing each lottery $f(s)$ by its canonical lottery with the same objective expected utility leaves the act indifferent, so every act is indifferent to its canonical utility act. The constant acts $\ell_0$ and $\ell_1$ bracket every canonical utility act by statewise monotonicity. Define
\(
U_f:=\{a\in[0,1]:\ell_a\succeq f\}
\;\text{and}\;
L_f:=\{a\in[0,1]:f\succeq\ell_a\}.
\)
Completeness implies $U_f\cup L_f=[0,1]$, the endpoint bracketing makes both sets nonempty, and continuity makes both closed. Since $[0,1]$ is connected, two nonempty closed sets whose union is the whole interval must intersect. Hence $f\sim\ell_v$ for some $v\in[0,1]$. Strict state monotonicity makes this $v$ unique; write it as $V(f)$. Finally, if $f^n\to f$, compactness of $[0,1]$ implies that every subsequence of $\{V(f^n)\}$ has a convergent further subsequence, say $V(f^{n_k})\to\bar v$. Since $f^{n_k}\sim\ell_{V(f^{n_k})}$ and preference has closed graph, $f\sim\ell_{\bar v}$; uniqueness of the canonical certainty equivalent gives $\bar v=V(f)$. Hence every convergent subsequence has the same limit and $V(f^n)\to V(f)$.
\end{proof}

\subsection{Nested local support}\label[appendix]{app:local-solvability}

\begin{proof}[\textbf{Proof of \Cref{lem:local-solvability}}]
Fix $v$, an open representation neighborhood $J_v\ni v$, and two distinct cells $A,B$. Choose a third cell $C$. Pick $r^-<v<r^+$ with $r^-,r^+\in J_v$. With every cell other than $C$ fixed at $v$, strict state monotonicity gives
\[
\widehat f_{(v,\ldots,r^-,\ldots,v)}\prec\ell_v
\quad\text{and}\quad
\widehat f_{(v,\ldots,r^+,\ldots,v)}\succ\ell_v.
\]
Because closed-graph continuity of the weak preference implies that the strict preference relation is open, the same strict comparisons hold when the prescribed $A$- and $B$-coordinates are sufficiently close to $v$. For every such pair $(b,c)$, continuity in the $C$-coordinate therefore gives a balancing value between $r^-$ and $r^+$, hence inside $J_v$. Intersect the finitely many neighborhoods obtained for all ordered cell pairs and shrink once more so that its closure lies in $J_v$. The resulting interval is $I_v$.
\end{proof}

\subsection{Within-cell transfers and local compensation measurement}\label[appendix]{app:additive}

\begin{proof}[\textbf{Proof of \Cref{lem:within-cell-connect}}]
Let $x^0=x$. If $x^k\neq y$, the equality $p\cdot x^k=p\cdot y$ implies that some coordinate $s$ lies above its target and some coordinate $t$ lies below. Put
\[
\lambda=\min\{p_s(x_s^k-y_s),\ p_t(y_t-x_t^k)\}>0
\]
and change the two coordinates by $x_s^{k+1}=x_s^k-\lambda/p_s$ and $x_t^{k+1}=x_t^k+\lambda/p_t$, leaving all others fixed. The $p$-mean is unchanged, each of the two coordinates remains between its current value and its target value, and at least one of $s,t$ reaches its target. A coordinate that has reached its target need not be moved again. Since the cell is finite, the procedure reaches $y$ after finitely many steps, and every intermediate vector stays coordinatewise between the current point and the target.
\end{proof}

\begin{proof}[\textbf{Proof of \Cref{lem:local-mean-suff}}]
Fix \(v\) and first consider one feasible two-state transfer segment inside a revealed cell that preserves the recovered \(p_{v,A}\)-mean and whose cell-mean vector lies in \(J_v^{\T}\). Suppose two points on that segment were ranked differently relative to \(\ell_v\). If one of them were indifferent to \(\ell_v\), A5 applied at that point would make every feasible point on the same mean-preserving segment indifferent to \(\ell_v\), contradicting the ranking of the other point. If instead one point were strictly preferred to \(\ell_v\) and the other strictly worse, continuity of \(V\) would yield an intermediate point indifferent to \(\ell_v\), and A5 would again imply that every feasible point on the segment is indifferent to \(\ell_v\), a contradiction. Hence the ranking relative to \(\ell_v\) is constant along every such segment.

Now take two canonical acts with the same vector of recovered cell means in \(J_v^{\T}\). Within one cell, \Cref{lem:within-cell-connect} joins their state-utility vectors by finitely many feasible two-state transfers that preserve that cell mean and never leave the coordinate box spanned by the endpoints. The vector of cell means is therefore unchanged and remains in \(J_v^{\T}\) throughout. The preceding argument preserves the ranking relative to \(\ell_v\) along every step. Apply the construction cell by cell. For arbitrary acts \(f\) and \(g\), A2 gives \(f\sim f_{u(f)}\) and \(g\sim f_{u(g)}\). If \(f\) and \(g\) have the same vector of recovered cell means, their two canonical utility acts do as well, so the preceding transfer argument applies to those canonical acts. Transferring the resulting comparison back through A2 proves that \(f\) and \(g\) have the same ranking relative to \(\ell_v\).
\end{proof}

\begin{proof}[\textbf{Proof of \Cref{lem:derived-order}}]
A8 makes any two directly measurable changes comparable after matched subdivision, so $\succeq_v^F$ is complete. For transitivity, suppose $I\succeq_v^FJ$ and $J\succeq_v^FK$ but $I\not\succeq_v^FK$. Completeness then gives $K\succ_v^FI$. Expand the matched decompositions and all direct chains supporting the three derived comparisons. Their endpoint-record differences add to
\(
(\chi_I-\chi_J)+(\chi_J-\chi_K)+(\chi_K-\chi_I)=0.
\)
At least one constituent direct comparison must be strict; otherwise every chain could be reversed, contradicting $K\succ_v^FI$. A7 therefore rules out the cycle. Hence $I\succeq_v^FK$, proving transitivity and therefore that $\succeq_v^F$ is a weak order.

For the direct-to-derived property, a one-piece decomposition gives $I\succeq_v^FJ$ whenever $I\succeq_v^{\mathrm{dir}}J$. If $I\succ_v^{\mathrm{dir}}J$ but also $J\succeq_v^FI$, expand the chains establishing the latter comparison. Their endpoint-record differences telescope to $\chi_J-\chi_I$. Adding the original strict direct comparison, whose difference is $\chi_I-\chi_J$, gives an A7 cancellation cycle with a strict comparison, a contradiction. Thus $I\succ_v^{\mathrm{dir}}J$ implies $I\succ_v^FJ$. If the direct verdict is indifference, the same one-piece argument in both orientations gives derived indifference.

All zero changes have the same endpoint record. A8 makes any two zero changes comparable. If one were strictly ranked above another, expand the chains supporting that strict derived comparison. The total endpoint-record difference is zero, while at least one constituent direct comparison must be strict; otherwise all chains could be reversed. This contradicts A7. Hence all zero changes are mutually indifferent and define the class $0$.

For reversal, suppose $I\succeq_v^FJ$ but $-J\not\succeq_v^F-I$. Completeness gives $-I\succ_v^F-J$. Expanding the two derived comparisons into constituent direct comparisons, the combined endpoint record of $(I,-I)$ equals that of $(J,-J)$, namely zero. A7 is therefore violated. Hence $I\succeq_v^FJ$ iff $-J\succeq_v^F-I$.

Finally, suppose $[a,b]_A\succeq_v^F[a',b']_B$ and $[b,c]_A\succeq_v^F[b',c']_B$. If $[a',c']_B\succ_v^F[a,c]_A$, expansion of the three comparisons gives a cancellation cycle because
\[
\chi_{[a,b]_A}+\chi_{[b,c]_A}=\chi_{[a,c]_A},
\qquad
\chi_{[a',b']_B}+\chi_{[b',c']_B}=\chi_{[a',c']_B}.
\]
This contradicts A7. Completeness gives the stated concatenation conclusion; if one premise is strict, the same argument rules out an indifferent conclusion. This proves all parts of the lemma.
\end{proof}

The derived order therefore supplies the qualitative algebra required for measurement. The next lemma provides the abstract difference-measurement result used to construct cardinal cell scales. The only external measurement-theoretic input in its proof is the within-cell algebraic-difference representation of \citet[Ch.~4, Sec.~4.4.1, Def.~3 and Thm.~2; proof in Sec.~4.5.2]{KLST1971}. Under the identification of an ordered pair \(ab\) with the tagged change \([a,b]_A\), KLST Axioms 1--5 correspond respectively to weak order, reversal, concatenation, interval solvability, and the bounded-standard-sequence condition. All cross-cell cardinal alignment is established separately in the proof. In the application below, in the proof of \Cref{lem:additive}, \(D_A=I_v\) for every cell \(A\); allowing cell-specific domains keeps the measurement lemma abstract.

\begin{lemma}\label{lem:measurement-bridge}
Consider finitely many cells. For each cell \(A\), let \(D_A\) be an open interval of admissible endpoint values, so that a tagged directed change in cell \(A\) is a pair \([a,b]_A\) with \(a,b\in D_A\). Let \(\succeq\) be a weak order on the collection of all such tagged directed changes.
 Suppose all zero changes are mutually indifferent and denote their class by $0$. Define standard and strictly bounded sequences exactly as above, with $D_A$ replacing $I_v$ and $\succeq$ replacing $\succeq_v^F$. Assume: (i) reversal; (ii) concatenation with preservation of strictness; (iii) interval solvability: for every cell $A$, every $a,b\in D_A$, and every tagged directed change $I$ (in any cell), whenever $[a,b]_A\succeq I\succeq0$, there are $x,y\in D_A$ such that $[a,x]_A\sim I\sim[y,b]_A$; (iv) no infinite standard sequence is strictly bounded; and (v) for every pair of cells $A,B$, the comparison set
\(
\{(a,b,c,d)\in D_A^2\times D_B^2:[a,b]_A\succeq[c,d]_B\}
\)
is relatively closed; and (vi) scalar orientation holds: for every cell $A$ and $a,b\in D_A$, $a>b$ if and only if $[a,b]_A\succ0$. Then there are continuous strictly increasing functions $\psi_A:D_A\to\mathbb R$ such that
\[
[a,b]_A\succeq[c,d]_B
\quad\Longleftrightarrow\quad
\psi_A(a)-\psi_A(b)\ge\psi_B(c)-\psi_B(d).
\]
The family is unique up to cell-specific additive constants and one common positive multiplicative factor.
\end{lemma}

\begin{proof}
Restrict the relation first to one cell $A$. Identify the ordered pair $ab\in D_A\times D_A$ in \citet[Ch.~4, Sec.~4.4.1, Def.~3 and Thm.~2; proof in Sec.~4.5.2]{KLST1971} with the tagged change $[a,b]_A$. KLST Axioms 1--5 are supplied respectively by the weak-order hypothesis, reversal, concatenation, interval solvability, and the exclusion of infinite strictly bounded standard sequences. Hence the algebraic-difference theorem yields a function $\varphi_A:D_A\to\mathbb R$, unique up to positive affine transformation, such that, for all $a,b,c,d\in D_A$,
\[
[a,b]_A\succeq[c,d]_A
\quad\Longleftrightarrow\quad
\varphi_A(a)-\varphi_A(b)
\ge
\varphi_A(c)-\varphi_A(d).
\]
This is the only step imported from KLST. Scalar orientation then selects the increasing orientation of the scale, while continuity is established from the present relative-closedness and solvability assumptions as follows. To make the continuity step explicit, fix $b\in D_A$ and vary the first endpoint $a$. Assumption (v) makes the weak upper and lower comparison sections of the changes $[a,b]_A$ relatively closed, while completeness makes the corresponding strict sections relatively open. Suppose the strictly increasing representative $\varphi_A$ had a left jump at some $a^*\in D_A$. A monotone real function has continuity points arbitrarily close to $a^*$. Hence one can choose $r>s$ in $D_A$, in the same cell $A$, so close to such a continuity point that the positive difference
\(
\varphi_A(r)-\varphi_A(s)
\)
is strictly smaller than the size of the jump. Put $I=[r,s]_A$. Choosing $b<a^*$ sufficiently close makes $[a^*,b]_A\succ I\succ0$. Interval solvability then requires an $x\in D_A$ with $[a^*,x]_A\sim I$, and therefore
\(
\varphi_A(a^*)-\varphi_A(x)
=
\varphi_A(r)-\varphi_A(s).
\)
Thus $\varphi_A(x)$ lies strictly between the left limit of $\varphi_A$ at $a^*$ and $\varphi_A(a^*)$, which is impossible for a strictly increasing function with a left jump at $a^*$. The symmetric solvability clause rules out right jumps in the same way. Hence $\varphi_A$ has neither left nor right jumps and is continuous on $D_A$. Relative closedness ensures that the limiting weak comparisons used in this argument remain valid. No cross-cell cardinal comparison is imported from the cited theorem. Write the resulting within-cell difference as
\(
d_A(a,b):=\varphi_A(a)-\varphi_A(b).
\)

It remains to place the cell scales on one common unit. We first record the elementary topological implication used below. Because $\succeq$ is complete and every cross-cell comparison set is relatively closed, its strict part is relatively open: if $I\succ J$ and sequences $I_n\to I$, $J_n\to J$ failed eventually to satisfy $I_n\succ J_n$, completeness would give a subsequence with $J_n\succeq I_n$, and relative closedness would imply $J\succeq I$, a contradiction.

Fix a reference cell $A_0$ and a nonzero reference change $P=[p^+,p^0]_{A_0}\succ0$. By this openness of strict preference, changes sufficiently close to a zero change in any cell lie strictly between $P$ and $-P$. Interval solvability in the reference cell therefore implies that, for every cell $A$, there is a neighborhood $N_A$ of zero such that each $x\in N_A$ can be realized by an $A$-change and matched by indifference to an $A_0$-change. For $x<0$ use reversal of the corresponding positive comparison.

For $x\in N_A$, choose an $A$-change $I$ with $d_A(I)=x$ and an $A_0$-change $J$ with $I\sim J$, and define
\(
h_A(x):=d_{A_0}(J).
\)
This is well defined. If $J$ and $J'$ match two $A$-changes having the same $d_A$-value, the within-cell representation makes those two $A$-changes indifferent; transitivity then gives $J\sim J'$, so $d_{A_0}(J)=d_{A_0}(J')$. Cross-cell order preservation similarly makes $h_A$ strictly increasing and gives $h_A(0)=0$.

The matching values approach zero with $x$. To see this, realize a sequence $x_n\to0$ by changes whose endpoints converge to the same interior point of $D_A$, and choose their reference matches inside the fixed reference interval between $p^0$ and $p^+$. If the represented reference differences failed to converge to zero, compactness of that fixed interval would give a subsequential limit equal to a nonzero reference change. Relative closedness would then make that nonzero change indifferent to a zero change, contradicting scalar orientation. Hence $h_A(x)\to0$ as $x\to0$.

Take sufficiently small $x,y,x+y\in N_A$. By openness of $D_A$ and continuity of $\varphi_A$, realize $x$ and $y$ by adjacent $A$-changes. Since $h_A(x)$ and $h_A(y)$ are also sufficiently small, openness of $D_{A_0}$ and continuity of $\varphi_{A_0}$ allow adjacent reference-cell changes with represented values $h_A(x)$ and $h_A(y)$. Each reference change is indifferent to the corresponding $A$-change by the definition and well-definedness of $h_A$. Concatenation in both directions therefore makes the combined changes indifferent, so
\(
h_A(x+y)=h_A(x)+h_A(y).
\)
Thus $h_A$ is strictly increasing and locally additive. The monotone Cauchy argument gives
\(
h_A(x)=\lambda_Ax
\)
on a neighborhood of zero for some $\lambda_A>0$. Put $\psi_A:=\lambda_A\varphi_A$ (and normalize the reference cell so that $\lambda_{A_0}=1$). We now verify that these calibrated scales represent \emph{all}, not only sufficiently small, cross-cell comparisons.

Take arbitrary changes $I=[a,b]_A$ and $J=[c,d]_B$, and write
\(
X=\psi_A(a)-\psi_A(b)
\;\text{and}\;
Y=\psi_B(c)-\psi_B(d).
\)
Choose a common integer $m$ large enough that $X/m$ and $Y/m$ lie in the respective locally calibrated ranges. Since each $\psi_A$ is continuous and strictly increasing on an interval, subdivide $I$ and $J$ monotonically into $m$ pieces $I_1,\ldots,I_m$ and $J_1,\ldots,J_m$ satisfying
\(
\Delta\psi_A(I_k)=X/m
\;\text{and}\;
\Delta\psi_B(J_k)=Y/m
\;\text{for every }k.
\)
Each small piece can be matched by indifference to a reference-cell change whose $\psi_{A_0}$-difference equals its calibrated value. Hence, if $X\ge Y$, the within-reference-cell difference representation and transitivity give $I_k\succeq J_k$ for every $k$; repeated concatenation yields $I\succeq J$. If $X<Y$, the same argument with the roles reversed gives $J\succ I$, using preservation of strictness under concatenation. Therefore
\(
I\succeq J
\;\Longleftrightarrow\;
\Delta\psi_A(I)\ge\Delta\psi_B(J),
\)
which establishes the asserted cross-cell representation on the full difference domains.

For uniqueness, let $\{\widetilde\psi_A\}_A$ be any other family representing the same cross-cell difference order. Within-cell uniqueness of the algebraic-difference representation gives, for each cell $A$,
\(
\widetilde\psi_A=\alpha_A+\rho_A\psi_A,
\; \rho_A>0.
\)
Fix the reference cell $A_0$. For any other cell $A$, choose a nonzero sufficiently small $A$-change $I$ and a reference-cell change $J$ with $I\sim J$; the construction above guarantees such a pair. Under the constructed common-unit family, the corresponding differences satisfy
\(
\Delta\psi_A(I)=\Delta\psi_{A_0}(J)\ne0.
\)
Because the alternative family represents the same indifference,
\(
\rho_A\,\Delta\psi_A(I)=\rho_{A_0}\,\Delta\psi_{A_0}(J).
\)
The common nonzero difference therefore implies $\rho_A=\rho_{A_0}$. Since $A$ was arbitrary, all cell-specific multipliers coincide with one common $\lambda>0$, while the origins $\alpha_A$ remain unrestricted because only differences are represented. Hence any alternative family has the form
\(
\widetilde\psi_A=\alpha_A+\lambda\psi_A
\)
with the same $\lambda$ for every cell, proving the stated common-unit uniqueness.
\end{proof}

\begin{proof}[\textbf{Proof of \Cref{lem:additive}}]
By \Cref{lem:derived-order}, A6--A8 give a weak order with a common zero class, reversal, and concatenation. A9 gives relative closedness, the interval-solvability condition in \Cref{lem:measurement-bridge}, and the bounded-standard-sequence condition. It remains to verify scalar orientation. If $a>b$, compare $[a,b]_A$ with a zero change in another cell $B$. If $A\triangleleft B$, the direct experiment raises the $A$-coordinate and leaves the $B$-coordinate unchanged, so A3 gives a strict improvement over $\ell_v$. If $B\triangleleft A$, the definition of $\succeq_v^{\mathrm{dir}}$ implements the reversed $A$-change, which is a strict worsening, and reverses the comparison with $\ell_v$. In either ordering, $[a,b]_A\succ_v^{\mathrm{dir}}0$. Then \Cref{lem:derived-order} gives $[a,b]_A\succ_v^F0$. Reversal gives the opposite ranking when $a<b$.

All hypotheses of \Cref{lem:measurement-bridge} therefore hold with $D_A=I_v$. We obtain continuous strictly increasing $\psi_{v,A}:I_v\to\mathbb R$ satisfying \eqref{eq:local-additive}. \Cref{lem:measurement-bridge} also gives the uniqueness result
\(
\widetilde\psi_{v,A}=\alpha_A+\lambda\psi_{v,A},
\; \lambda>0,
\)
with the same $\lambda$ for every cell. Subtracting $\psi_{v,A}(v)$ normalizes each scale to zero at $v$ without changing represented differences.
\end{proof}

\begin{proof}[\textbf{Proof of \Cref{lem:derived-exact-balance}}]
Fix directly measurable changes $I=[a,b]_A$ and $J=[c,d]_B$ in distinct cells with $I\sim_v^FJ$. By \Cref{lem:local-solvability}, the forward compensation experiment is supportable. Suppose first that $A\triangleleft B$. For any supporting $v$-indifferent background with coordinates $(b,c)$, the same post-change block act determines both orientations of the direct comparison: $I\succeq_v^{\mathrm{dir}}J$ exactly when it is weakly above $\ell_v$, while $J\succeq_v^{\mathrm{dir}}I$ exactly when it is weakly below $\ell_v$. If the post-change act were strictly above $\ell_v$, then $I\succ_v^{\mathrm{dir}}J$; if it were strictly below, then $J\succ_v^{\mathrm{dir}}I$. Either strict direct ranking would imply the same strict derived ranking by \Cref{lem:derived-order}(i), contradicting $I\sim_v^FJ$. Hence the post-change act is indifferent to $\ell_v$, so the pair is locally exactly balanced.

Now suppose $B\triangleleft A$. By the definition of the direct order, both comparisons are oriented through the reverse compensation experiment: start from a supporting $v$-indifferent block vector with coordinates $(a,d)$ and move to $(b,c)$. If the terminal act were strictly below $\ell_v$, then $I\succ_v^{\mathrm{dir}}J$; if it were strictly above, then $J\succ_v^{\mathrm{dir}}I$. Either conclusion would contradict $I\sim_v^FJ$ through \Cref{lem:derived-order}(i). The terminal vector with coordinates \((b,c)\) is therefore itself \(v\)-indifferent. It can consequently serve as a supporting background for the compensation experiment defining local exact balance for the pair \((I,J)\): applying \(I\) in cell \(A\) and the reverse of \(J\) in cell \(B\) moves the coordinates from \((b,c)\) to \((a,d)\), which is precisely the original \(v\)-indifferent vector. Hence \((I,J)\) is locally exactly balanced in this ordering as well.
\end{proof}

\begin{proof}[\textbf{Proof of \Cref{lem:common-scale}}]
Choose a reference cell $A_0$ and put $\phi_v=\psi_{v,A_0}$. For any cell $A$, define $F_A=\psi_{v,A}\circ\phi_v^{-1}$ on the interval $\phi_v(I_v)$. A10 says that $F_A$ preserves the ordering of all objective differences. Hence equal reference-scale differences have equal $F_A$-differences. Therefore $F_A(x+t)-F_A(x)$ depends only on $t$ whenever the displayed points remain in $\phi_v(I_v)$. The induced increment map is locally additive and continuous, so it is linear near zero with some slope $\beta_{v,A}>0$. Any larger feasible difference can be subdivided into sufficiently small pieces; summing the local identities shows that the same slope applies throughout the interval. Thus
\(
F_A(x)=\alpha_{v,A}+\beta_{v,A}x.
\)
After normalizing each $\psi_{v,A}$ and $\phi_v$ to vanish at $v$, $\alpha_{v,A}=0$, proving $\psi_{v,A}=\beta_{v,A}\phi_v$.
\end{proof}

\subsection{A finite balancing argument}\label[appendix]{app:finite-transfer}

The following elementary lemma converts an aggregate zero condition into a finite sequence of pairwise weighted cancellations; it will be used in the local, global-DI, and stable-attitude characterization proofs.

\begin{lemma}\label{lem:balancing-schedule}
Let $q_1,\ldots,q_K>0$ and let $y_1,\ldots,y_K\in\mathbb R$ satisfy $\sum_iq_i y_i=0$. Starting from zero, one can reach $y$ in finitely many steps, each step changing one positive remaining component and one negative remaining component so that their weighted changes cancel, with each changed component remaining between its current value and its target value.
\end{lemma}

\begin{proof}
At any unfinished step choose one positive and one negative residual and move the pair until one residual becomes zero while preserving their weighted sum. At least one additional coordinate is completed at each step, so the procedure ends after finitely many steps.
\end{proof}

\section{Proofs for the local and global doubly implicit characterizations}\label[appendix]{app:di-baseline-proofs}

\subsection{Local doubly implicit characterization}\label[appendix]{app:baseline}

\begin{proof}[\textbf{Proof of \Cref{prop:local-baseline}}]
We prove the two directions separately.

\emph{A1--A11 imply regular neutrally separated local DI.}
A1--A3 give the normalized vNM index $u$, statewise canonicalization, strict state monotonicity, and the continuous certainty-equivalent index $V$. A4 gives the fixed revealed partition and, through \eqref{eq:p-from-rates}, the full-support within-cell beliefs $p_{v,A}$.

Fix \(v\). By \Cref{lem:local-mean-suff}, whenever the vector of cell means lies in \(J_v^{\T}\), an act can be replaced by the block act with the same vector of recovered conditional means without changing its ranking relative to \(\ell_v\).
By \Cref{lem:additive,lem:common-scale}, A6--A10 give the directly measured scale $\phi_v:I_v\to\mathbb R$, the positive coefficients $\beta_{v,A}$, and the normalized recovered weights $q_v(A)=\beta_{v,A}/\sum_C\beta_{v,C}$. By A11(ii), choose $\Theta_v\in\mathcal E_v^{L}$; by definition $\Theta_v=\phi_v$ on $I_v$.

We first characterize the $v$-indifferent block acts in $J_v^{\T}$. Use the auxiliary compact core $K_v\subset I_v$ fixed above, with $v\in\operatorname{int}K_v$. Suppose
\(
\sum_Aq_v(A)\Theta_v(z_A)=\Theta_v(v).
\)
Apply \Cref{lem:balancing-schedule} to the targets $y_A=\Theta_v(z_A)-\Theta_v(v)$. Each scheduled step changes two cells, say $A$ and $B$, by $\Theta_v$-increments $\Delta_A$ and $\Delta_B$ satisfying
\(
q_v(A)\Delta_A+q_v(B)\Delta_B=0.
\)
Because $v\in\operatorname{int}K_v$ and $\Theta_v=\phi_v$ on $K_v$, the set of differences generated inside $K_v$ contains a neighborhood of zero. Choose a common integer $N$ large enough that both $|\Delta_A|/N$ and $|\Delta_B|/N$ lie in that difference range, and subdivide the two scheduled coordinate moves into $N$ monotone pieces having respective $\Theta_v$-increments $\Delta_A/N$ and $\Delta_B/N$. Each paired substep therefore preserves weighted cancellation:
\(
q_v(A)\frac{\Delta_A}{N}
+
q_v(B)\frac{\Delta_B}{N}
=0.
\)

Because $\Theta_v=\phi_v$ on $I_v$, and hence on $K_v$, choose core changes in cells $A$ and $B$ with these same respective differences. Their weighted measured values cancel, so \eqref{eq:local-additive} makes them indifferent under $\sim_v^F$. By \Cref{lem:derived-exact-balance}, the core pair is locally exactly balanced. Write the scheduled coordinate moves in the same orientation as $I_k=[x_k,x_{k-1}]_A$ and $-J_k=[y_k,y_{k-1}]_B$, so that $J_k=[y_{k-1},y_k]_B$ records the compensating change. The current $v$-indifferent block vector supplies the starting coordinates $(x_{k-1},y_{k-1})$, hence the scheduled pair is locally supportable. A11(ii) transports the core balance to $(I_k,J_k)$; A11(i) makes the verdict independent of the other cell coordinates throughout $J_v$. Every coordinate remains between its value in $v\one$ and its value in $z$, so every intermediate block vector remains in the coordinate box spanned by $v\one$ and $z$ and therefore lies in $J_v^{\T}$. Iterating the finite schedule reaches $z$ while preserving indifference.

If the displayed sum is strictly above $\Theta_v(v)$, choose $c\in J_v$ below $v$ and below every coordinate of $z$. The aggregate at $c\one$ equals $\Theta_v(c)<\Theta_v(v)$. Continuity along the segment from $c\one$ to $z$ gives a block vector $z^0\le z$, strict in at least one coordinate, with aggregate exactly $\Theta_v(v)$. The zero case gives $\widehat f_{z^0}\sim\ell_v$, and strict state monotonicity gives $\widehat f_z\succ\ell_v$. The case of an aggregate below $\Theta_v(v)$ is symmetric, using a constant block vector above both $v$ and all coordinates of $z$. Hence, for every $z\in J_v^{\T}$ and every pair $(R,r)\in\{(\succ,>),(\sim,=),(\prec,<)\}$,
\[
\widehat f_z\,R\,\ell_v
\quad\Longleftrightarrow\quad
\sum_Aq_v(A)\Theta_v(z_A)\,r\,\Theta_v(v).
\]
Local within-cell mean sufficiency now gives \eqref{eq:local-di} for arbitrary acts whose cell-mean vector lies in $J_v^{\T}$.

It remains to verify local neutral separation. If distinct revealed cells $A,B$ and states $s\in A$, $t\in B$ satisfied the local representation-side two-sided neutral equality for all sufficiently small $\varepsilon$, the local block formula at the constant background $v\one$ would make the corresponding two-cell block perturbation indifferent to $\ell_v$. \Cref{lem:local-mean-suff} would transfer that verdict to the two-state perturbation with the same vector of cell means, yielding $s\sim_v^0t$, contrary to $A\ne B$. Thus the local representation is neutrally separated.

\emph{Regular neutrally separated local DI implies A1--A11.}
Clause (i) of \Cref{def:local-di}, continuity of $V$, and the vNM restriction on constants give A1; statewise utility invariance gives A2; clause (ii) gives A3. If $s,t$ belong to one represented cell $A$, the rate $p_{v,A}(s)/p_{v,A}(t)$ leaves the cell mean unchanged, so the local threshold formula at the constant act $\ell_v$ makes every sufficiently small two-sided transfer at that rate indifferent to $\ell_v$. If $s,t$ lie in distinct represented cells, local neutral separation rules out such a rate. Because the represented partition is fixed across $v$ and contains at least three cells, the behavioral neutral-link relation coincides with it. The probability ratios automatically satisfy the multiplicative rate identities required by A4; reciprocity follows separately from the two-sided neutral-transfer definition and uniqueness under A3.

For A5, a within-cell transfer at the revealed rate leaves the complete vector of represented cell means unchanged. Whenever that vector lies in $J_v^{\T}$, the local threshold formula therefore preserves indifference. On $I_v$, put
\(
W_v([a,b]_A)=q_v(A)[\Theta_v(a)-\Theta_v(b)].
\)
The local DI formula gives
\(
I\succeq_v^{\mathrm{dir}}J
\;\Longleftrightarrow\;
W_v(I)\ge W_v(J),
\)
independently of the untouched supporting coordinates; hence A6 holds. If a finite list of direct comparisons has the same endpoint records on the two sides, additivity of $W_v$ makes the two total represented changes equal. Therefore none of the component inequalities can be strict, giving A7.

For A8, directly measurable changes in distinct cells are directly comparable through the displayed $W_v$ order. If two changes $I,J$ lie in the same cell, choose a third cell. Suppose first that $W_v(I)\ge W_v(J)$. Choose a common integer $m$ large enough that $W_v(I)/m$ and $W_v(J)/m$ both lie in the neighborhood of zero generated by sufficiently small third-cell changes. By continuity and strict monotonicity of $\Theta_v$, subdivide $I$ into $m$ monotone pieces each having represented value $W_v(I)/m$, and subdivide $J$ into $m$ monotone pieces each having value $W_v(J)/m$. For each matched pair choose a third-cell change whose represented value lies between these two values. Each $I$-piece is then linked by a direct chain to the corresponding $J$-piece with the same orientation. The case $W_v(J)\ge W_v(I)$ is symmetric. Hence matched decompositions with a common orientation exist, giving A8. The same construction shows that the derived order is exactly the $W_v$ order. Continuity of $\Theta_v$ gives the closedness part of A9; intermediate-value solvability gives its solvability clause; and a nonzero standard sequence has represented accumulated differences growing linearly in the number of steps, so it cannot remain strictly between two fixed finite bounds. Hence A9 holds.

Same-cell comparisons cancel the positive factor $q_v(A)$, so the objective-difference order is the same in every cell, giving A10. By the local DI formula, the sign of every supportable two-cell compensation in $J_v$ depends only on the two weighted $\Theta_v$-differences. This gives A11(i) directly on all of $J_v$; restricting it to endpoints in $I_v$ gives A6. The scale recovered on $I_v$ is unique up to a positive affine transformation; normalize it so that $\phi_v=\Theta_v$ on $I_v$. Then $\Theta_v\in\mathcal E_v(\phi_v;I_v,J_v)$, and equal $\Theta_v$-differences in the same two cells preserve local exact balance, giving A11(ii). Thus A1--A11 hold.
\end{proof}

\subsection{Global behavioral calibratability}\label[appendix]{app:gbc}

\begin{lemma}\label{lem:di-global-balance-suff}
Suppose A1--A12 hold and let $\Phi_v$ be supplied by A12(iii). If a balance-feasible pair $I=[a,b]_A$, $J=[c,d]_B$ satisfies
\(
q_v(A)[\Phi_v(a)-\Phi_v(b)]
=
q_v(B)[\Phi_v(c)-\Phi_v(d)],
\)
then the pair is globally exactly balanced.
\end{lemma}

\begin{proof}
Keep the ordered pair $(I,J)$ and its original starting coordinates $(b,c)$ throughout. Put
\[
\Delta_A=\Phi_v(a)-\Phi_v(b),
\qquad
\Delta_B=\Phi_v(c)-\Phi_v(d),
\]
so that $q_v(A)\Delta_A=q_v(B)\Delta_B$. Because A12(iii) requires $\Phi_v=\phi_v$ on $I_v$, the difference range generated inside the auxiliary compact core $K_v$ fixed above contains a neighborhood of zero. Choose a common integer $m$ large enough that $\Delta_A/m$ and $\Delta_B/m$ both lie in that range.

Define the actual coordinate paths by
\[
\Phi_v(x_k)=\Phi_v(b)+\frac{k}{m}\Delta_A,
\qquad
\Phi_v(y_k)=\Phi_v(c)-\frac{k}{m}\Delta_B,
\qquad k=0,\ldots,m.
\]
Thus $x_0=b$, $x_m=a$, $y_0=c$, and $y_m=d$. At step $k$ write
\[
I_k=[x_k,x_{k-1}]_A,
\qquad
J_k=[y_{k-1},y_k]_B.
\]
The joint compensation experiment applies $I_k$ and the reverse of $J_k$, so it moves the actual coordinates from $(x_{k-1},y_{k-1})$ to $(x_k,y_k)$. Moreover,
\[
q_v(A)\bigl[\Phi_v(x_k)-\Phi_v(x_{k-1})\bigr]
=
q_v(B)\bigl[\Phi_v(y_{k-1})-\Phi_v(y_k)\bigr].
\]
For every $k$, choose core changes $\bar I_k$ in cell $A$ and $\bar J_k$ in cell $B$, with endpoints inside $K_v$, having these same respective $\Phi_v$-differences. Since $\Phi_v=\phi_v$ on $K_v$, the two core changes have equal measured weighted values and hence are indifferent under $\sim_v^F$. By \Cref{lem:derived-exact-balance}, each core pair is locally, and hence globally, exactly balanced.

Because the original pair is balance-feasible, choose a supporting $v$-indifferent block vector whose $A$- and $B$-coordinates are $(x_0,y_0)=(b,c)$. Inductively, suppose the current $v$-indifferent block vector has coordinates $(x_{k-1},y_{k-1})$. It therefore makes $(I_k,J_k)$ balance-feasible. A12(iii) transports the global exact balance of $(\bar I_k,\bar J_k)$ to $(I_k,J_k)$, and A12(ii) makes that verdict independent of the particular supporting background. After applying $I_k$ together with $-J_k$, the new block vector remains $v$-indifferent and has coordinates $(x_k,y_k)$, so it supports the next step. Iterating to $k=m$ proves that the original pair is globally exactly balanced.
\end{proof}

\begin{proof}[\textbf{Proof of \Cref{thm:global-di-baseline}}]
Assume A1--A12 and fix $v$. Let $\Phi_v$ be supplied by A12(iii). By \Cref{lem:di-global-balance-suff}, every balance-feasible pair with cancelling weighted $\Phi_v$-differences is globally exactly balanced. Put
\(
H_v(z)=\sum_Aq_v(A)\Phi_v(z_A).
\)
If $H_v(z)=\Phi_v(v)$, apply \Cref{lem:balancing-schedule} to the targets $\Phi_v(z_A)-\Phi_v(v)$, starting from $v\one$. Translating each scheduled increment back through the strictly increasing $\Phi_v$ gives a finite sequence of pairwise coordinate changes that keeps every coordinate between its current value and its target value. Inductively, the current block vector is $v$-indifferent, so the next scheduled pair is balance-feasible on that current background. Its two weighted $\Phi_v$-increments cancel, and \Cref{lem:di-global-balance-suff} therefore makes the post-change block vector $v$-indifferent as well. The finite schedule reaches $z$, hence $\widehat f_z\sim\ell_v$.

If $H_v(z)>\Phi_v(v)$, consider the segment from $0\one$ to $z$. Since $H_v(0\one)=\Phi_v(0)<\Phi_v(v)$, continuity gives $z^0\le z$, strict in at least one coordinate, with $H_v(z^0)=\Phi_v(v)$. The zero case gives $\widehat f_{z^0}\sim\ell_v$ and A3 gives $\widehat f_z\succ\ell_v$. The case $H_v(z)<\Phi_v(v)$ is symmetric, using the segment from $z$ to $1\one$. Thus the global threshold formula holds for every block act.

It remains to pass from block acts to arbitrary acts. A12(i) gives within-cell mean invariance on the full domain. The same continuity argument used in \Cref{lem:local-mean-suff} shows that the ranking relative to \(\ell_v\) is constant along each feasible mean-preserving transfer segment: otherwise an indifferent intermediate act would exist and A12(i) would make the entire segment indifferent. \Cref{lem:within-cell-connect} then implies that any two acts with the same vector of recovered conditional means have the same ranking relative to \(\ell_v\). Replacing an arbitrary act by the block act with the same cell means therefore gives \eqref{eq:global-di}. Because \(\Phi_v\) agrees with the directly measured scale on \(I_v\), the local neutral-separation property already obtained under A1--A11 applies to sufficiently small cross-cell perturbations and supplies the neutral-separation clause of the global representation. This is the global DI representation.

Conversely, a regular neutrally separated global DI representation restricts to a regular neutrally separated local DI representation on every admissible \(J_v\), so \Cref{prop:local-baseline} gives A1--A11. The global formula makes within-cell mean-preserving transfers neutral everywhere, giving A12(i). For any two supporting backgrounds, the post-change act's ranking relative to \(\ell_v\) depends only on the same two weighted \(\Phi_v\)-differences and is therefore independent of the supporting background, giving A12(ii). On \(I_v\), the recovered directly measured attitude scale is positive-affinely equivalent to \(\Phi_v|_{I_v}\); normalize it so that they coincide. Equal \(\Phi_v\)-differences in the same two cells preserve global exact balance throughout the full domain, giving A12(iii). Hence all three clauses of A12 hold.
\end{proof}

\section{Proofs for the stability refinements}\label[appendix]{app:stability-proofs}

\subsection{Cross-level scale alignment}\label[appendix]{app:cross-level-scale-alignment}

\begin{proof}[\textbf{Proof of \Cref{lem:cross-level-scale-alignment}}]
If $I_v\cap I_w\neq\varnothing$, A13 makes the directly measured objective-difference orders coincide on the overlap. By the uniqueness of the difference measures recovered under A6--A10, there are unique constants $\alpha_{vw}\in\mathbb R$ and $\gamma_{vw}>0$ such that
\(
\phi_w=\alpha_{vw}+\gamma_{vw}\phi_v
\quad\text{on }I_v\cap I_w.
\)
The transition is unique because a nonempty intersection of two open intervals contains more than one point.

Any two measurement neighborhoods can be connected by a finite chain of pairwise-overlapping $I$-intervals; otherwise the union of the intervals connected to a fixed base point and the union of the remaining intervals would separate $(0,1)$. Fix one base scale and propagate its normalization along such a chain. We show that the resulting normalization does not depend on the chosen chain.

It is enough to check closed overlap chains. For intervals on the real line, any finite overlap cycle of length at least four has a chord between two nonconsecutive intervals; repeated use of such chords reduces the cycle to triangles. For three pairwise-overlapping open intervals, the interval property implies a nonempty common intersection. On that common overlap, both the direct transition from the first scale to the third and the composition of the first-to-second and second-to-third transitions transform the same nonconstant function into the third scale. Uniqueness of the positive-affine transition therefore makes the direct and composed maps identical. Hence the transition around every triangle, and therefore every finite closed chain, is the identity.

The propagated normalizations are consequently path independent. The normalized local scales agree on every overlap and define one continuous strictly increasing function $\phi:(0,1)\to\mathbb R$. Changing the initial normalization changes every normalized local scale by the same positive affine transformation, giving uniqueness of the common positive-affine class.
\end{proof}

\subsection{Global stable-attitude characterization}\label[appendix]{app:main-proof}

\begin{lemma}\label{lem:global-difference}
Suppose A1--A14 hold. Let $\phi:(0,1)\to\mathbb R$ be the common attitude scale supplied by A13. Then there is a unique continuous strictly increasing extension, denoted $\phi:[0,1]\to\mathbb R$, such that
\(
\delta(a,b)=\phi(a)-\phi(b)
\;(a,b\in[0,1]).
\)
\end{lemma}

\begin{proof}
By A14(i), the interior difference function $\delta^\circ(a,b)=\phi(a)-\phi(b)$ has a continuous real-valued extension $\delta$ to $[0,1]^2$. On the interior,
\(
\delta^\circ(a,c)=\delta^\circ(a,b)+\delta^\circ(b,c).
\)
Approximating any boundary triple by interior triples and using continuity extends this identity to all $a,b,c\in[0,1]$. Fix $r\in(0,1)$ and define
\(
\bar\phi(a)=\phi(r)+\delta(a,r),\; a\in[0,1].
\)
Additivity gives
\(
\delta(a,b)=\delta(a,r)-\delta(b,r)=\bar\phi(a)-\bar\phi(b).
\)
For interior $a$, $\delta(a,r)=\phi(a)-\phi(r)$, so $\bar\phi$ agrees with the previously constructed $\phi$ on $(0,1)$. Continuity of $\delta$ makes $\bar\phi$ continuous on $[0,1]$.

The extension is strictly increasing. Interior strict monotonicity is inherited directly. If $0<a<1$, choose $0<x_0<c<a$. Continuity gives $\bar\phi(0)=\lim_{x\downarrow0}\phi(x)\le\phi(x_0)$, while strict interior monotonicity gives $\phi(x_0)<\phi(c)<\phi(a)=\bar\phi(a)$; hence $\bar\phi(0)<\bar\phi(a)$. The argument at 1 is symmetric. Any other continuous real-valued extension agreeing with the interior scale has the same endpoint limits, so the extension is unique. We henceforth write it as $\phi$.
\end{proof}

\begin{lemma}\label{lem:stable-global-balance-suff}
Suppose A1--A14 hold. If a balance-feasible pair $I=[a,b]_A$, $J=[c,d]_B$ satisfies
\(
q_v(A)[\phi(a)-\phi(b)]
=
q_v(B)[\phi(c)-\phi(d)],
\)
then the pair is globally exactly balanced at level $v$.
\end{lemma}

\begin{proof}
Put $\Delta_A=\phi(a)-\phi(b)$ and $\Delta_B=\phi(c)-\phi(d)$. Since $I_v$ is an open neighborhood of $v$ and the common interior attitude scale is a positive-affine normalization of the directly measured level-$v$ scale on $I_v$, the difference set generated by $\phi(I_v)$ contains a neighborhood of zero. Choose a common integer $m$ so large that $\Delta_A/m$ and $\Delta_B/m$ both lie in that difference set.

Define $x_0=b,x_m=a$ and $y_0=c,y_m=d$ by
\[
\phi(x_k)=\phi(b)+\frac{k}{m}\Delta_A,
\qquad
\phi(y_k)=\phi(c)-\frac{k}{m}\Delta_B,
\qquad k=0,\ldots,m.
\]
At step $k$ set $I_k=[x_k,x_{k-1}]_A$ and $J_k=[y_{k-1},y_k]_B$. Thus the joint experiment applies $I_k$ and $-J_k$, moving the actual coordinates from $(x_{k-1},y_{k-1})$ to $(x_k,y_k)$, and the weighted $\phi$-differences cancel. Choose directly measurable core changes $\bar I_k,\bar J_k$ in $I_v$ with the same respective $\phi$-differences. On $I_v$ the common scale differs from the measured level-$v$ scale only by one positive affine transformation, so equality of the $q_v$-weighted common-scale differences gives $\bar I_k\sim_v^F\bar J_k$. By \Cref{lem:derived-exact-balance}, each core pair is locally, and hence globally, exactly balanced.

Start from a supporting $v$-indifferent background for the original balance-feasible pair. If the current background has coordinates $(x_{k-1},y_{k-1})$, then $(I_k,J_k)$ is balance-feasible. A14(ii) transports global exact balance from $(\bar I_k,\bar J_k)$ to $(I_k,J_k)$ because the two component $\delta$-differences coincide. A12(ii) makes the resulting verdict hold on the current supporting background. The resulting block vector is again $v$-indifferent and supports the next step. Finite iteration reaches $(a,d)$ and proves global exact balance of the original pair.
\end{proof}

\begin{proof}[\textbf{Proof of \Cref{thm:main}}]
Assume A1--A14. By \Cref{thm:global-di-baseline}, choice already has a global DI representation. A13 and \Cref{lem:cross-level-scale-alignment} give one common ambiguity-attitude scale on the directly measured neighborhoods; A14(i) and \Cref{lem:global-difference} extend it to $[0,1]$.

Fix $v$. Suppose
\(
\sum_Aq_v(A)\phi(z_A)=\phi(v).
\)
Apply \Cref{lem:balancing-schedule} to the targets $\phi(z_A)-\phi(v)$, starting from the constant block vector $v\one$. At each scheduled step two coordinates move with cancelling $q_v$-weighted $\phi$-increments. Write the first coordinate move as a change $I_k$ and the second actual move as the reverse of a change $J_k$. The current block vector is $v$-indifferent, so $(I_k,J_k)$ is balance-feasible, and the cancelling weighted increments give the hypothesis of \Cref{lem:stable-global-balance-suff}. That lemma therefore preserves $v$-indifference at every step. Because every coordinate remains between its initial and target values and the schedule terminates after finitely many moves, it reaches $z$, proving $\widehat f_z\sim\ell_v$.

The strict cases follow by continuity and strict state monotonicity, exactly as in the proof of \Cref{thm:global-di-baseline}. If the aggregate is above $\phi(v)$, continuity along the segment from $0\one$ to $z$ gives $z^0\le z$, strict in at least one coordinate, with aggregate exactly $\phi(v)$; hence $\widehat f_{z^0}\sim\ell_v$ and A3 gives $\widehat f_z\succ\ell_v$. If the aggregate is below $\phi(v)$, the symmetric argument along the segment from $z$ to $1\one$ gives $z^1\ge z$, strict in at least one coordinate, on the zero contour, and therefore $\widehat f_z\prec\ell_v$. Thus, for every block vector $z$ and every pair $(R,r)\in\{(\succ,>),(\sim,=),(\prec,<)\}$,
\[
\widehat f_z\,R\,\ell_v
\quad\Longleftrightarrow\quad
\sum_Aq_v(A)\phi(z_A)\,r\,\phi(v).
\]
A12(i) extends the revealed conditional-mean reduction to arbitrary acts, yielding the stable-attitude representation in \Cref{def:iisa}.

Conversely, a global representation with one common $\phi$ is a special case of global DI and therefore gives A1--A12. Its objective-difference order is the same at every level, giving A13; its common difference scale has the required continuous boundary extension and replacing a globally exactly balanced pair by one with the same two $\phi$-differences preserves global exact balance, giving A14. Hence A13 and A14 hold.
\end{proof}

\subsection{Standard identifiable smooth-ambiguity characterization}\label[appendix]{app:stable-perception}

\begin{proof}[\textbf{Proof of \Cref{thm:standard-sa}}]
Assume A1--A15. By \Cref{thm:main}, the preference has the global stable-attitude representation with one common $\phi:[0,1]\to\mathbb R$. A15(i) fixes the within-cell beliefs. For a singleton cell this is immediate. If $|A|\ge2$, then for every pair of distinct states $s,t\in A$ and every pair of levels,
\(
\frac{p_{v,A}(s)}{p_{v,A}(t)}
=
\frac{p_{w,A}(s)}{p_{w,A}(t)}.
\)
Equality of all pairwise ratios together with normalization gives $p_{v,A}=p_A$.

It remains to fix the across-cell weights. Take $v,w$ with $U=I_v\cap I_w\neq\varnothing$ and two cells $A,B$. Because $U$ is open and $\phi$ is continuous and strictly increasing, its difference set contains a neighborhood of zero. Choose sufficiently small nonzero differences $x=\phi(a)-\phi(b)$ and $y=\phi(c)-\phi(d)$, with $a,b,c,d\in U$, satisfying
\(
q_v(A)x=q_v(B)y.
\)
At level $v$ the corresponding directly measured changes are indifferent. A15(ii) transfers that indifference to level $w$, where the common global $\phi$ represents the same objective differences. Hence
\(
q_w(A)x=q_w(B)y.
\)
Since $x,y\ne0$,
\(
\frac{q_v(A)}{q_v(B)}
=
\frac{q_w(A)}{q_w(B)}.
\)
This holds for every pair of cells on every overlap. The finite-overlap-chain argument proved in \Cref{lem:cross-level-scale-alignment} applies to the open cover $\{I_v\}_{v\in(0,1)}$, so the ratio equalities propagate across all levels. Normalization gives one common vector $q$.

Therefore
\[
H(f)=\sum_{A\in\T}q(A)\phi\!\left(\sum_{s\in A}p_A(s)u(f(s))\right)
\]
is independent of the certainty-equivalent level. Setting $v=V(f)$ in the global threshold representation gives $H(f)=\phi(V(f))$ whenever $V(f)\in(0,1)$. If $V(f)=0$, the canonical utility vector of $f$ must be $0\one$: otherwise strict state monotonicity would give $f\succ\ell_0$, contradicting $f\sim\ell_0$. Hence $H(f)=\phi(0)$. Similarly, $V(f)=1$ forces the canonical utility vector to be $\one$ and therefore $H(f)=\phi(1)$. Thus $V(f)=\phi^{-1}(H(f))$ on the full domain, which is the standard identifiable smooth-ambiguity representation.

Conversely, suppose the standard fixed-perception representation holds. It is a special case of the global DI baseline, so \Cref{thm:global-di-baseline} yields A1--A12, and it is a stable-attitude representation, so \Cref{thm:main} yields A13 and A14. The neutral rate within cell $A$ is $p_A(s)/p_A(t)$ at every level, giving A15(i). The directly measured cross-cell order at every level is represented by
\[
q(A)[\phi(a)-\phi(b)]\ge q(B)[\phi(c)-\phi(d)],
\]
which contains no level-specific object, so A15(ii) holds. Hence A15 follows.
\end{proof}

\subsection{Alternative orders of refinement}\label[appendix]{app:alternative-refinements}

\begin{proof}[\textbf{Proof of \Cref{prop:gbc-a12}}]
Under A1--A12, \Cref{thm:global-di-baseline} supplies at least one global DI representation with level-specific indices $\Phi_v$. If A13 holds, \Cref{lem:cross-level-scale-alignment} gives one common scale on the directly measured neighborhoods. Positive-affine renormalization of each $\Phi_v$ does not change the level-$v$ threshold representation, so each can be normalized to agree with that common scale on $I_v$. Conversely, if the global indices admit such normalizations, the same common objective-difference comparison represents every level on overlaps, which is A13.
\end{proof}

\begin{proof}[\textbf{Proof of \Cref{prop:perception-local}}]
Assume A15. For the within-cell perception margin, the conclusion is immediate on a singleton cell, whose conditional belief is necessarily degenerate. If $|A|\ge2$, A15(i) and uniqueness of neutral-transfer rates under A3 give $r_{st}^v=r_{st}^w$ for every pair of distinct states in $A$ and every pair of levels. Since
\(
r_{st}^v=\frac{p_{v,A}(s)}{p_{v,A}(t)},
\)
equality of all pairwise ratios together with normalization implies $p_{v,A}=p_{w,A}$. Thus for every cell we may write the common conditional belief as $p_A$.

Now fix $v,w$ with $U:=I_v\cap I_w\neq\varnothing$. On $U$, the level-$v$ and level-$w$ derived compensation orders are represented by
\(
\psi_{v,A}=\beta_{v,A}\phi_v
\;\text{and}\;
\psi_{w,A}=\beta_{w,A}\phi_w.
\)
A15(ii) says that every cross-cell comparison generated by these systems is the same at the two levels. For each cell $A$, let $X_{v,A}$ be the set of level-$v$ represented $\psi_{v,A}$-differences generated by intervals in $U$. Since $U$ is open and there are finitely many cells, choose a symmetric neighborhood $N$ of zero contained in every $X_{v,A}$; shrink it if necessary so that sufficiently small sums can be represented by adjacent intervals inside $U$.

For $x\in N$ and any cell $A$, choose an $A$-interval whose level-$v$ represented difference is $x$, and define $h_A(x)$ as its level-$w$ represented difference. This is well defined. If two $A$-intervals have the same level-$v$ value $x$, choose in a distinct cell an interval with that same value. Both cross-cell comparisons are indifferent at level $v$ and therefore, by A15(ii), at level $w$; the two level-$w$ values must coincide. The same argument across two cells gives $h_A(x)=h_B(x)$, so there is one map $h:N\to\mathbb R$. Cross-cell order invariance makes $h$ strictly increasing: if $x<y$, realize the two values in distinct cells and transfer their strict comparison from level $v$ to level $w$.

Take $x,y,x+y\in N$ small enough to be represented by adjacent intervals in one cell. Additivity of numerical differences at both levels gives
\(
h(x+y)=h(x)+h(y).
\)
A strictly increasing additive function on a neighborhood of zero is linear there, so $h(x)=\lambda_{vw}x$ for some $\lambda_{vw}>0$. For an arbitrary interval in $U$, subdivide its level-$v$ represented difference into sufficiently small monotone pieces in $N$. Applying the local proportionality piece by piece and adding level-$w$ differences gives
\begin{equation}\label{eq:perception-common-unit}
\beta_{w,A}[\phi_w(a)-\phi_w(b)]
=
\lambda_{vw}\beta_{v,A}[\phi_v(a)-\phi_v(b)]
\qquad(a,b\in U).
\end{equation}
Choose $a\ne b$ and compare this identity across two cells. The common attitude differences cancel, yielding
\(
\frac{\beta_{w,A}}{\beta_{w,B}}
=
\frac{\beta_{v,A}}{\beta_{v,B}}.
\)
After normalization, $q_w=q_v$ on every overlap. Thus A15(ii) first fixes the ratios of the cross-cell coefficients and hence the normalized cell weights. The open neighborhoods $\{I_v\}_{v\in(0,1)}$ cover the connected interval $(0,1)$, so equality propagates along finite chains of overlaps; write the common vector as $q$. Only after this normalization step do we return to \eqref{eq:perception-common-unit}, which gives
\(
\phi_w(a)-\phi_w(b)
=
\gamma_{vw}[\phi_v(a)-\phi_v(b)]
\)
on every overlap for some $\gamma_{vw}>0$. Hence A13 holds. This proves that A15 implies fixed $\{p_A\}_{A\in\T}$, fixed $q$, and A13.

Conversely, suppose the within-cell beliefs and normalized cell weights are fixed across levels and A13 holds. Fixed $p_A$ gives the same neutral rate $p_A(s)/p_A(t)$ at every level, hence A15(i). On an overlap, A13 permits the two directly measured attitude scales to be positively affinely normalized to the same scale. Because the normalized weight vector is also the same, the represented cross-cell compensation order is identical at the two levels; common positive multipliers in the unnormalized $\beta$ coefficients do not affect the ordering. Hence A15(ii) holds. This proves the converse.
\end{proof}

\begin{proof}[\textbf{Proof of \Cref{cor:perception-gbc}}]
If A12 and A15 hold, \Cref{thm:global-di-baseline} gives global DI and \Cref{prop:perception-local} fixes $\{p_A\}_{A\in\T}$ and $q$ and gives A13. By \Cref{prop:gbc-a12}, the global attitude indices can therefore be normalized to agree with one common scale on the directly measured neighborhoods. Conversely, such a global representation gives A12 through the global-DI theorem and A15 through fixed perception plus the common measured cross-cell ordering.
\end{proof}

\subsection{Local stable attitude}\label[appendix]{app:local-stable-attitude}

\begin{proof}[\textbf{Proof of \Cref{prop:local-stable}}]
If A13$^*$ holds, its A13 component and \Cref{lem:cross-level-scale-alignment} give one common positive-affine attitude class on the directly measured neighborhoods. For each $v$, the representative $\phi^v$ agrees with the directly measured scale $\phi_v$ on $I_v$, and A13$^*$(ii) is exactly the local common-scale replacement property required in the definition of local validity. Hence
\(
\phi^v|_{J_v}\in\mathcal E_v^{L}.
\)
The balancing argument in the proof of \Cref{prop:local-baseline} applies to any such locally valid index, so $\phi^v$ generates the local threshold representation throughout $J_v$. Because the threshold comparison is unchanged by a positive affine transformation of the level-$v$ attitude index, any single representative $\phi$ of the common class can therefore be used in every local threshold formula, giving \eqref{eq:local-stable}. Conversely, a local representation with one common $\phi$ makes the directly measured difference order independent of $v$ on every overlap, giving A13; after level-$v$ affine normalization, the common index agrees with $\phi_v$ on $I_v$ and equal common-scale differences preserve local exact balance on $J_v$, so its restriction belongs to $\mathcal E_v^{L}$ and A13$^*$(ii) follows.
\end{proof}

\begin{proof}[\textbf{Proof of \Cref{cor:perception-local-common}}]
By \Cref{prop:perception-local}, A15 fixes $\{p_A\}_{A\in\T}$ and $q$ and implies A13. Adding the strengthening required by A13$^*$ lets the common directly measured attitude scale represent every local domain $J_v$, so \Cref{prop:local-stable} gives the fixed-perception local stable-attitude representation. The converse is immediate from the same two propositions.
\end{proof}

\section{Identification and uniqueness}\label[appendix]{app:identification}
This appendix proves the two layers of construction invariance collected in \Cref{lem:family-invariance}: first, construction invariance of the level-specific perception objects and of the locally revealed attitude class in a neighborhood of each level under the global DI baseline; second, the additional uniqueness of the common global ambiguity-attitude class on the stable-attitude path.

\subsection{Construction invariance across elicitation constructions}

\begin{proof}[\textbf{Proof of \Cref{lem:family-invariance}}]
\emph{Part (i): fixed-level identification.}
The partition $\T$ and the conditional probabilities $p_{v,A}$ are constructed from the primitive neutral-transfer relation and its unique rates. They are therefore independent of the local-solvability domains, the orientation order, and the arbitrary reference state used to recover within-cell probabilities. It remains to compare the weights.

Fix a level $v$. Let the two constructions use measurement neighborhoods $I_v$ and $\widetilde I_v$, and put $U_v=I_v\cap\widetilde I_v$, an open interval containing $v$. Because both global DI representations describe the same certainty-equivalent index $V$, their zero contours agree on $U_v^{\T}$. Centering the two global indices at $v$ gives, for every $z\in U_v^{\T}$,
\begin{equation}\label{eq:baseline-family-zero-set}
\sum_A q_v(A)[\Phi_v(z_A)-\Phi_v(v)]=0
\quad\Longleftrightarrow\quad
\sum_A \widetilde q_v(A)[\widetilde\Phi_v(z_A)-\widetilde\Phi_v(v)]=0.
\end{equation}

For each cell $A$, let
\(
X_A=q_v(A)[\Phi_v(U_v)-\Phi_v(v)],
\)
an open interval containing zero, and for $y\in X_A$ define
\[
g_A(y)
:=
\widetilde q_v(A)
\left[
\widetilde\Phi_v\!\left(
\Phi_v^{-1}\!\left(\Phi_v(v)+\frac{y}{q_v(A)}\right)
\right)
-
\widetilde\Phi_v(v)
\right].
\]
Each $g_A$ is continuous, strictly increasing, and satisfies $g_A(0)=0$. Equation \eqref{eq:baseline-family-zero-set} becomes, on the product of the $X_A$'s,
\(
\sum_A y_A=0
\;\Longleftrightarrow\;
\sum_A g_A(y_A)=0.
\)
Because there are finitely many cells and every \(X_A\) is an open interval containing zero, choose \(\eta>0\) such that \((-2\eta,2\eta)\subseteq X_A\) for every cell \(A\). Fix three distinct cells \(A,B,C\). For \(|x|,|y|<\eta\),
$
g_A(x)+g_B(y)+g_C(-x-y)=0.
$

Setting \(y=0\) and then \(x=0\) gives
$
g_A(t)=-g_C(-t)
\;\text{and}\;
g_B(t)=-g_C(-t)
$
for all sufficiently small \(t\), so \(g_A=g_B\) locally. Setting \(y=-x\) in the three-cell identity gives
$
g_A(x)+g_B(-x)=0.
$
Hence the common local map \(g_A=g_B\) is odd. Combining this with
$
g_A(t)=-g_C(-t)
$
shows that \(g_C=g_A=g_B\) locally. Denote this common local map by \(g\). Substituting back into the three-cell identity yields
$
g(x+y)=g(x)+g(y)
$
whenever the relevant arguments lie in the common neighborhood. Continuity therefore implies
$
g(t)=\lambda_v t
$
locally for some \(\lambda_v>0\).

Finally, pairing any remaining cell \(D\) with \(C\) in a two-coordinate zero-sum vector \((t,-t)\) shows that \(g_D(t)=\lambda_v t\) locally as well. Hence the same formula holds for every cell.
Hence, after shrinking to a common neighborhood $N_v\ni v$ if necessary,
\begin{equation}\label{eq:baseline-family-local-affine}
\widetilde q_v(A)[\widetilde\Phi_v(a)-\widetilde\Phi_v(v)]
=
\lambda_v q_v(A)[\Phi_v(a)-\Phi_v(v)]
\qquad(a\in N_v)
\end{equation}
for every cell $A$. Comparing \eqref{eq:baseline-family-local-affine} for two cells at the same $a\ne v$ gives
\(
\frac{\widetilde q_v(A)}{\widetilde q_v(B)}
=
\frac{q_v(A)}{q_v(B)}.
\)
Both vectors are normalized to sum to one, so $\widetilde q_v=q_v$. Substituting this equality back into \eqref{eq:baseline-family-local-affine} and cancelling the common positive factor $q_v(A)$ gives
\[
\widetilde\Phi_v(a)-\widetilde\Phi_v(v)
=
\lambda_v[\Phi_v(a)-\Phi_v(v)]
\qquad(a\in N_v),
\]
which is the local attitude conclusion in part~(i). Since $v$ was arbitrary, the weights and these locally revealed attitude classes are construction invariant at every level. Together with construction invariance of $p_{v,A}$, this behaviorally identifies $\pi_v(s)=q_v(A)p_{v,A}(s)$ and therefore, by \Cref{lem:finite-ident}, behaviorally identifies $(P_v,\mu_v)$ up to relabelling of cells (equivalently, models). The argument does not link the local affine normalizations across levels or establish uniqueness of the global continuations.

\medskip
\emph{Part (ii): stable-attitude identification.}
By part~(i), the two constructions already recover the same $p_{v,A}$ and $q_v$ at every level. It remains to compare their common stable-attitude indices $\phi$ and $\widetilde\phi$.

Fix $v$ and put $U_v=I_v\cap\widetilde I_v$. Applying the fixed-level argument from part~(i) to the two stable-attitude representations, and using equality of the weight vectors, yields an open interval $N_v\ni v$ and constants $\alpha_v\in\mathbb R$, $\lambda_v>0$ such that
\[
\widetilde\phi(a)=\alpha_v+\lambda_v\phi(a)
\qquad(a\in N_v).
\]
The intervals $\{N_v\}_{v\in(0,1)}$ cover the connected interval $(0,1)$. Whenever two overlap, the two positive-affine formulas relate the same two nonconstant functions on an open interval, so their slopes and intercepts coincide. Connectedness propagates this equality across the cover. Hence there are one $\lambda>0$ and one $\alpha\in\mathbb R$ such that
\(
\widetilde\phi=
\alpha+\lambda\phi
\;\text{on }(0,1).
\)
Both stable-attitude indices are continuous on $[0,1]$, so the same affine equality extends to the endpoints by one-sided limits. This proves the claim.
\end{proof}

The following example illustrates the representational nonuniqueness left open by \Cref{lem:family-invariance}: a preference may admit an identified common stable-attitude class while also admitting alternative level-specific DI indices outside that class. \Cref{sec:identification,cor:identification} record the resulting identification conclusions; the role of the example is specifically to show that A12-compatible full-domain continuations may remain nonunique even though A13--A14 behaviorally identify one common global attitude class.

\begin{example}\label{ex:nonunique-global-continuation}
The failure of uniqueness of $\Phi_v$ under the global DI baseline can be genuine. Let
$S=\{s_1,s_2,s_3\}$ and let the revealed partition consist of the three singleton cells
$A_i=\{s_i\}$. Give every cell weight $q_v(A_i)=1/3$ and write
$x_i=u(f(s_i))\in[0,1]$. Consider the preference represented by
\[
V(f)
=
\left(\frac{x_1^2+x_2^2+x_3^2}{3}\right)^{1/2}.
\]
It has the common stable-attitude index $\phi(x)=x^2$, with singleton conditional beliefs and equal cell weights. The representation is full support and neutrally separated: for distinct singleton cells, a candidate two-sided neutral rate $r>0$ around a constant $v$ would require
\[
(v+\varepsilon)^2+(v-r\varepsilon)^2-2v^2
=
2v(1-r)\varepsilon+(1+r^2)\varepsilon^2
=0
\]
for all sufficiently small feasible $\varepsilon$, which is impossible. The displayed regular stable-attitude representation directly satisfies A1--A4, so \Cref{lem:local-solvability} supplies admissible nested neighborhoods and hence an admissible elicitation construction. Relative to any such construction, the converse direction of \Cref{thm:global-di-baseline} gives A1--A12 and the converse direction of \Cref{thm:main} gives A13--A14.

At every level $v$, one valid centered DI index is
\(
\Phi_v(x)=x^2-v^2
\;\text{with}\;
\Phi_v(v)=0.
\)
Fix $v_0<1/\sqrt{3}$ and put $M=\sqrt{3}v_0<1$. Shrink the directly measured neighborhood at $v_0$, if necessary, so that $I_{v_0}\subset(0,M)$. Every $v_0$-indifferent block vector satisfies
\[
x_1^2+x_2^2+x_3^2=3v_0^2=M^2,
\]
and therefore every one of its coordinates is at most $M$. Values above $M$ never occur on the $v_0$-indifference surface.

For any $\eta>0$, define a second continuous strictly increasing continuation by
\[
\widetilde\Phi_{v_0}(x)
=
\begin{cases}
 x^2-v_0^2, & 0\le x\le M,\\[1mm]
 2v_0^2+\eta(x-M), & M<x\le1.
\end{cases}
\]
The two indices agree on $[0,M]$, and hence on $I_{v_0}$, but are not positive-affine transformations of one another on $[0,1]$: because they coincide on a nondegenerate interval, any global positive-affine relation would have to be the identity, whereas their upper tails differ.

Nevertheless they generate exactly the same $v_0$-threshold relation. If every coordinate satisfies $x_i\le M$, the two aggregate indices coincide. If some $x_i>M$, then
\[
\widetilde\Phi_{v_0}(x_i)>2v_0^2,
\qquad
\widetilde\Phi_{v_0}(x_j)\ge -v_0^2
\quad(j\ne i),
\]
so
\(
\sum_j\widetilde\Phi_{v_0}(x_j)>0
\).
At the same time $x_i^2>3v_0^2$, so
\(
\sum_j\Phi_{v_0}(x_j)>0
\).
Thus, for every block vector $x\in[0,1]^3$,
\[
\sum_i\Phi_{v_0}(x_i)\gtreqless0
\quad\Longleftrightarrow\quad
\sum_i\widetilde\Phi_{v_0}(x_i)\gtreqless0.
\]

The alternative continuation is also compatible with the replacement requirement in A12(iii). Every globally exactly balanced pair at level $v_0$ starts and ends on the $v_0$-indifference surface, so all of its relevant endpoints lie in $[0,M]$, where $\Phi_{v_0}$ and $\widetilde\Phi_{v_0}$ have identical differences. If a balance-feasible replacement has those same $\widetilde\Phi_{v_0}$-differences, the two weighted changes cancel from a $v_0$-indifferent starting background. The post-change aggregate under $\widetilde\Phi_{v_0}$ is therefore again zero, and the threshold equivalence above makes the resulting act $v_0$-indifferent. Keeping $\Phi_v(x)=x^2-v^2$ at every $v\ne v_0$ therefore gives a second A12-compatible family of global DI continuations for the same preference.

The example does not conflict with stable-attitude identification. The common index $\phi(x)=x^2$ works simultaneously at every certainty-equivalent level and is the structural attitude object identified, up to one common positive affine transformation, by A13--A14 and part~(ii) of \Cref{lem:family-invariance}. In particular, after expressing this common class in the level-$v$ normalization, $\phi^v(x)=x^2-v^2$ is valid on the full domain and therefore on every local representation neighborhood $J_v$; hence A13$^*$ also holds. The modified $\widetilde\Phi_{v_0}$ does not contradict A13$^*$: that condition requires the common class itself to provide a locally valid index on every $J_v$, not that every admissible level-specific DI continuation belong to the common class. Indeed, keeping the canonical quadratic continuation at all other levels yields a global DI family whose level-specific indices are not all positive-affine transformations of one another, even though the same preference admits the stable-attitude representation with common class $[x^2]$. Hence A13--A14 identify the common cross-level structural class $[\phi]$; they do not make every separately admissible level-specific continuation $\Phi_v$ unique.
\end{example}

\section{Examples separating the attitude-stability restrictions}\label[appendix]{app:attitude-examples}

The next two examples complement \Cref{ex:nonunique-global-continuation}. That example concerns representational nonuniqueness even when global stable attitude, and hence A13$^*$, holds. Here the purpose is different: to separate failure of directly measured cross-level alignment from failure of globalization. The first example violates A13 itself; the second satisfies the stronger local condition A13$^*$ but violates A14(ii).

\begin{example}\label{ex:a13-failure}
Let $S=\{s_1,s_2,s_3\}$, let the revealed partition consist of the three singleton cells $A_i=\{s_i\}$, and set $q_v(A_i)=1/3$ for every $v$ and $i$. Write $x_i=u(f(s_i))\in[0,1]$. Fix $\eta\in(0,1)$ and define the level-specific indices
\(
\Phi_v(x)=x+\eta v x^2.
\)
For each $x=(x_1,x_2,x_3)$, write
\[
\bar x:=\frac13\sum_{i=1}^3x_i,
\qquad
\overline{x^2}:=\frac13\sum_{i=1}^3x_i^2,
\]
and put
\[
H_x(v)
:=
\frac13\sum_{i=1}^3\Phi_v(x_i)-\Phi_v(v)
=
\bar x+\eta v\,\overline{x^2}-v-\eta v^3.
\]
Since
\[
\frac{\partial H_x(v)}{\partial v}
=
\eta\overline{x^2}-1-3\eta v^2
\le \eta-1<0,
\]
while $H_x(0)\ge0$ and $H_x(1)\le0$, there is a unique $V(x)\in[0,1]$ satisfying $H_x(V(x))=0$. Hence
\[
V(x)\gtreqless v
\quad\Longleftrightarrow\quad
\frac13\sum_{i=1}^3\Phi_v(x_i)
\gtreqless
\Phi_v(v).
\]
Since $H_x(v)$ is jointly continuous in $(x,v)$, strictly decreasing in $v$, and has a unique zero, the root $V(x)$ depends continuously on $x$. It fixes constants. Moreover, increasing any coordinate $x_i$ strictly raises $H_x(v)$ for every $v$, so the unique zero $V(x)$ strictly increases; hence $V$ is strictly state-monotone. The singleton representation has full support. It is also neutrally separated: a two-sided neutral rate $r>0$ between two distinct singleton cells around a constant $v$ would require
\(
(1+2\eta v^2)(1-r)\varepsilon
+
\eta v(1+r^2)\varepsilon^2
=0
\)
for all sufficiently small feasible $\varepsilon$, which is impossible. Thus the displayed global DI representation directly satisfies A1--A4. Hence \Cref{lem:local-solvability} supplies admissible nested neighborhoods and therefore an admissible elicitation construction. Fix any such construction. Relative to this construction, the converse direction of \Cref{thm:global-di-baseline} implies A1--A12.

A13, however, fails. At level $v$ the directly measured objective differences are represented locally by
\(
D_v(a,b)
=
\Phi_v(a)-\Phi_v(b)
=
(a-b)\bigl[1+\eta v(a+b)\bigr].
\)
Take distinct sufficiently close levels $v\ne w$ for which $I_v\cap I_w$ is a nondegenerate interval. If A13 held, the two algebraic-difference orders would coincide on that overlap, so positive-affine uniqueness of the local difference representation would imply
\(
x+\eta w x^2
=
\alpha+\lambda\bigl(x+\eta v x^2\bigr)
\)
throughout an open interval, for some $\lambda>0$. Equality of polynomial coefficients gives $\alpha=0$, $\lambda=1$, and $w=v$, a contradiction. Hence the global DI preference satisfies A1--A12 but fails A13 and therefore admits no global stable-attitude representation. The failure is behavioral: the locally revealed attitude-difference order itself changes across certainty-equivalent levels.
\end{example}

\begin{example}\label{ex:a14-transport-failure}
Again take three singleton cells with equal weights $q_v(A_i)=1/3$. Let
\[
\phi(x)=x+x^2,
\qquad
\rho=\frac14,
\qquad
\kappa=\frac12,
\qquad
h(t)=(t-\rho)_+^2,
\qquad y_+:=\max\{y,0\},
\]
and define
\[
\Phi_v(x)
=
\phi(x)+\kappa v(1-v)h(x-v).
\]
Because $h(0)=0$, $\Phi_v(v)=\phi(v)$. For a block vector $x$, set
\(
H_x(v)
:=
\frac13\sum_i\Phi_v(x_i)-\Phi_v(v).
\)
Since \(x_i-v\in[-1,1]\), we have
$ 0\le h(t)\le \frac{9}{16}, \; 0\le h'(t)\le \frac32 \;\text{for all }t\in[-1,1], $
and \(v(1-v)\le1/4\), so
\[
\frac{\partial H_x(v)}{\partial v}
=
-(1+2v)
+
\frac12\left[
(1-2v)\frac13\sum_{i=1}^3 h(x_i-v)
-
v(1-v)\frac13\sum_{i=1}^3 h'(x_i-v)
\right] < -1+\frac{15}{32}<0.
\]
Together with $H_x(0)\ge0$ and $H_x(1)\le0$, this gives a unique zero $V(x)$. Joint continuity of $H_x(v)$ in $(x,v)$ and strict monotonicity in $v$ make $V$ continuous, while increasing any coordinate $x_i$ strictly raises $H_x(v)$ for every $v$ and therefore strictly raises the unique zero. Thus $V$ is strictly state-monotone and satisfies
\(
V(x)\gtreqless v
\;\Longleftrightarrow\;
\frac13\sum_i\Phi_v(x_i)
\gtreqless\Phi_v(v).
\)
For sufficiently small transfers around a constant \(v\), the perturbation term vanishes, so neutral separation is determined by the quadratic scale \(\phi\): a candidate rate \(r>0\) would have to satisfy
$
(1+2v)(1-r)\varepsilon+(1+r^2)\varepsilon^2=0
$
for all sufficiently small feasible \(\varepsilon\) of either sign. No \(r>0\) can satisfy this identity, since its quadratic coefficient \(1+r^2\) is strictly positive. Hence the representation is neutrally separated. Together with the regularity established above and the full-support perception specification, the representation is therefore regular, full support, and neutrally separated.
 It directly satisfies A1--A4, so before invoking the converse global characterization fix an admissible elicitation construction as follows. For every $v$, choose an open local representation neighborhood $J_v\ni v$ sufficiently small that
\(
J_v\subset (0,v+\rho)\cap(0,1),
\)
and apply \Cref{lem:local-solvability} to choose a corresponding directly measured neighborhood $I_v$ with $\overline I_v\subset J_v$. Relative to this construction, the converse of \Cref{thm:global-di-baseline} gives A1--A12. Since $h(x-v)=0$ throughout $J_v$,
\(
\Phi_v(x)=\phi(x)=x+x^2
\; \text{for all} \;x\in J_v.
\)
Thus the same common attitude class governs the entire local representation domain at every level. More precisely, in the level-$v$ normalization the representative
\(
\phi^v(x):=\phi(x)-\phi(v)
\)
agrees with the normalized directly measured scale on $I_v$ and is a locally valid index throughout $J_v$. Hence A13$^*$ holds, and therefore A13 holds as well. The common interior difference scale also extends continuously to the boundary as
\(
\delta(a,b)=\phi(a)-\phi(b),
\)
so A14(i) holds.

A14(ii) nevertheless fails. Fix $v_0=1/2$ and $\Delta=1/20$. Since the perturbation starts only above $3/4$, choose
\[
b=\frac12,
\qquad
c=\frac9{20},
\qquad
\phi(a)=\phi(b)+\Delta,
\qquad
\phi(d)=\phi(c)-\Delta.
\]
All four endpoints lie below $3/4$. Moreover the starting coordinates $(b,c)$ can be completed by a third coordinate $z<3/4$ satisfying
\(
\phi(b)+\phi(c)+\phi(z)=3\phi(v_0).
\)
Because $\Phi_{v_0}=\phi$ below $3/4$, the pair $I=[a,b]_{A_1}$, $J=[c,d]_{A_2}$ is balance-feasible and globally exactly balanced, with
\(
\delta(a,b)=\delta(c,d)=\Delta.
\)

Now consider the replacement endpoints
\[
b'=\frac45,
\qquad
c'=\frac3{10},
\qquad
\phi(a')=\phi(b')+\Delta,
\qquad
\phi(d')=\phi(c')-\Delta.
\]
The replacement is balance-feasible. Indeed,
\[
3\phi(v_0)-\Phi_{v_0}(b')-\Phi_{v_0}(c')
=
\frac{1343}{3200}
\in\bigl(0,\phi(3/4)\bigr),
\]
so a unique third coordinate $z'<3/4$ completes $(b',c')$ to a $v_0$-indifferent starting block vector. By construction,
\(
\delta(a',b')=\delta(a,b)=\Delta
\;\text{and}\;
\delta(c',d')=\delta(c,d)=\Delta.
\)
But $a'>b'>3/4$, whereas $c',d'<3/4$. Therefore
\[
\Phi_{v_0}(a')-\Phi_{v_0}(b')
>
\phi(a')-\phi(b')
=
\Delta,
\]
while
\[
\Phi_{v_0}(c')-\Phi_{v_0}(d')
=
\phi(c')-\phi(d')
=
\Delta.
\]
The two actual level-$v_0$ changes no longer cancel, so the replacement pair is not globally exactly balanced. This violates A14(ii). Hence even A13$^*$ together with A14(i) does not suffice for global stable attitude: the same common ambiguity-attitude class can govern every local representation neighborhood $J_v$ while failing to organize compensation on the full domain.
\end{example}

Together, \Cref{ex:a13-failure,ex:a14-transport-failure} separate three levels of attitude stability. In \Cref{ex:a13-failure}, even the directly measured attitude-difference order varies across certainty-equivalent levels, so A13 fails. In \Cref{ex:a14-transport-failure}, the same common attitude class is valid throughout every local representation neighborhood $J_v$, so A13$^*$ holds, but that class does not govern full-domain compensation because A14(ii) fails. Finally, \Cref{ex:nonunique-global-continuation} shows that even when the common attitude class is globally valid, and therefore A13$^*$ and A14 hold, alternative level-specific DI continuations outside that class may remain observationally equivalent. Thus the examples distinguish genuine local attitude variation, failure of local stability to globalize, and purely representational nonuniqueness.

\end{document}